\documentclass[11pt,letterpaper]{article}
\usepackage[lmargin=1.0in,rmargin=1.0in,bottom=1.0in,top=1.0in,twoside=False]{geometry}

\usepackage[utf8]{inputenc}
\usepackage[T1]{fontenc}
\usepackage{lmodern}

\usepackage{fullpage,amssymb,amsmath}
\usepackage{graphicx}
\usepackage{enumerate}
\usepackage{paralist}
\usepackage[inline]{enumitem}
\usepackage{tikz}
\usetikzlibrary{fit}
\usetikzlibrary{intersections}
\usetikzlibrary{decorations.pathreplacing}
\usetikzlibrary{calc}
\tikzset{every loop/.style={}}

\usepackage{xcolor}
\usepackage{mathtools}
\usepackage{microtype}
\usepackage{amsfonts}
\usepackage{comment}
\usepackage[english]{babel}
\usepackage{mathrsfs}
\usepackage[ruled,vlined]{algorithm2e}
\usepackage{blindtext}
\usepackage{thmtools}
\usepackage{subcaption}

\usepackage{trimspaces}
\usepackage{nccfoots}
\usepackage{setspace}
\usepackage{inconsolata}
\usepackage{libertine}
\usepackage[absolute]{textpos}

\usepackage{enumitem}
\usepackage{todonotes}
\usepackage{longtable}

\definecolor{blue}{rgb}{0.1,0.2,0.5}
\definecolor{brown}{rgb}{0.6,0.6,0.2}
\usepackage[ocgcolorlinks, linkcolor={blue}, citecolor={brown}]{hyperref}
\usepackage[amsmath,amsthm,thmmarks,hyperref]{ntheorem}
\usepackage{enumerate}
\usepackage{latexsym}

\usepackage{cleveref}

\crefformat{page}{#2page~#1#3}%
\Crefformat{page}{#2Page~#1#3}%
\crefformat{equation}{#2(#1)#3}%
\Crefformat{equation}{#2(#1)#3}%
\crefformat{figure}{#2Figure~#1#3}%
\Crefformat{figure}{#2Figure~#1#3}%
\crefformat{section}{#2Section~#1#3}
\Crefformat{section}{#2Section~#1#3}
\crefformat{chapter}{#2Chapter~#1#3}
\Crefformat{chapter}{#2Chapter~#1#3}
\crefformat{chapter*}{#2Chapter~#1#3}
\Crefformat{chapter*}{#2Chapter~#1#3}
\crefformat{part}{#2Part~#1#3}
\Crefformat{part}{#2Part~#1#3}
\crefformat{enumi}{#2(#1)#3}
\Crefformat{enumi}{#2(#1)#3}

\usepackage[mathlines]{lineno}

\newcommand*\patchAmsMathEnvironmentForLineno[1]{%
  \expandafter\let\csname old#1\expandafter\endcsname\csname #1\endcsname
  \expandafter\let\csname oldend#1\expandafter\endcsname\csname end#1\endcsname
  \renewenvironment{#1}%
     {\linenomath\csname old#1\endcsname}%
     {\csname oldend#1\endcsname\endlinenomath}}%
\newcommand*\patchBothAmsMathEnvironmentsForLineno[1]{%
  \patchAmsMathEnvironmentForLineno{#1}%
  \patchAmsMathEnvironmentForLineno{#1*}}%
\AtBeginDocument{%
  \patchBothAmsMathEnvironmentsForLineno{equation}%
  \patchBothAmsMathEnvironmentsForLineno{align}%
  \patchBothAmsMathEnvironmentsForLineno{flalign}%
  \patchBothAmsMathEnvironmentsForLineno{alignat}%
  \patchBothAmsMathEnvironmentsForLineno{gather}%
  \patchAmsMathEnvironmentForLineno{split}
  \patchBothAmsMathEnvironmentsForLineno{multline}}

\theoremnumbering{arabic}
\theoremstyle{plain}
\theoremsymbol{}
\theorembodyfont{\itshape}
\theoremheaderfont{\normalfont\bfseries}
\theoremseparator{.}

\newtheorem{theorem}{Theorem}
\crefformat{theorem}{#2Theorem~#1#3}
\Crefformat{theorem}{#2Theorem~#1#3}

\newcommand{\newtheoremwithcrefformat}[2]{%
  \newtheorem{#1}[theorem]{#2}%
  \crefformat{#1}{##2\MakeUppercase#1~##1##3}%
  \Crefformat{#1}{##2\MakeUppercase#1~##1##3}%
}
\newcommand{\newseptheoremwithcrefformat}[2]{%
  \newtheorem{#1}{#2}%
  \crefformat{#1}{##2\MakeUppercase#1~##1##3}%
  \Crefformat{#1}{##2\MakeUppercase#1~##1##3}%
}

\newtheoremwithcrefformat{lemma}{Lemma}

\newtheoremwithcrefformat{proposition}{Proposition}
\newtheoremwithcrefformat{observation}{Observation}
\newseptheoremwithcrefformat{conjecture}{Conjecture}
\newtheoremwithcrefformat{corollary}{Corollary}
\newtheorem*{claim*}{Claim}

\theoremstyle{definition}
\newtheorem*{example*}{Example}
\theorembodyfont{\upshape}

\theoremstyle{nonumberplain}
\theoremheaderfont{\scshape}
\theorembodyfont{\normalfont}
\theoremsymbol{\ensuremath{\square}}

\theoremsymbol{\ensuremath{\lrcorner}}

\renewcommand{\qed}{\hfill$\square$}

\newcommand{\siz}{{$n^2+32n+4m+2k$}}
\newtheorem{claimm}{Claim}

\newenvironment{inproof}{\noindent {\emph{Proof of Claim.}}}{\hfill$\blacksquare$\smallskip

}

\newtheorem{innercustomtheorem}{Theorem}
\newenvironment{customtheorem}[1]
  {\renewcommand\theinnercustomtheorem{#1}\innercustomtheorem}
  {\endinnercustomtheorem}

\renewcommand{\phi}{\varphi}
\renewcommand{\epsilon}{\varepsilon}

\renewcommand{\leq}{\leqslant}
\renewcommand{\geq}{\geqslant}

\newcommand{\sat}{\textsc{Sat}\xspace}
\newcommand{\is}{\textsc{Independent Set}\xspace}
\newcommand{\maxcut}{\textsc{Max Cut}\xspace}
\newcommand{\maxbis}{\textsc{Max Bisection}\xspace}

\newcommand{\whomo}[1]{\textsc{WHom}(#1)\xspace}

\newcommand{\lihomo}[1]{\textsc{LIHom}(#1)\xspace}
\newcommand{\lbhomo}[1]{\textsc{LBHom}(#1)\xspace}
\newcommand{\lshomo}[1]{\textsc{LSHom}(#1)\xspace}

\newcommand{\kdir}[1]{#1-\textsc{DIR}\xspace}
\newcommand{\pnaesat}{\textsc{PosNAE 3-Sat}\xspace}
\newcommand{\oct}{\textsc{Odd Cycle Transversal}\xspace}
\newcommand{\ioct}{\textsc{Independent Odd Cycle Transversal}\xspace}
\newcommand{\tos}{\xrightarrow{s}}

\newcommand{\cS}{\mathcal{S}}

\newcommand{\prop}{property $(\star )$\xspace}
\newcommand{\R}{\mathbb{R}}

\begin{document}
\title{Subexponential algorithms for variants of homomorphism problem in string graphs\thanks{This work was partially supported by Polish National Science Centre grant no. 2018/31/D/ST6/00062. The extended abstract of this paper was presented at the conference WG 2019~\cite{WG}.}}

\author{
Karolina Okrasa\thanks{
Faculty of Mathematics and Information Science, Warsaw University of Technology, Poland, \texttt{k.okrasa@mini.pw.edu.pl}}
\and
Pawe\l{}~Rz\k{a}\.zewski\thanks{
Faculty of Mathematics and Information Science, Warsaw University of Technology, Poland, \texttt{p.rzazewski@mini.pw.edu.pl}}\;\;\thanks{Corresponding author.}
}

\begin{titlepage}
\def\thepage{}
\thispagestyle{empty}
\maketitle

\begin{abstract}
We consider the complexity of finding weighted homomorphisms from intersection graphs of curves (string graphs) with $n$ vertices to a fixed graph $H$. We provide a complete dichotomy for the problem: if $H$ has no two vertices sharing two common neighbors, then the problem can be solved in time $2^{O(n^{2/3} \log n)}$, otherwise there is no algorithm working in time $2^{o(n)}$, even in intersection graphs of segments, unless the ETH fails.
This generalizes several known results concerning the complexity of computational problems in geometric intersection graphs.

Then we consider two variants of graph homomorphism problem, called locally injective homomorphism and locally bijective homomorphism, where we require the homomorphism to be injective or bijective on the neighborhood of each vertex. We show that for each target graph $H$, both problems can always be solved in time $2^{O(\sqrt{n} \log n)}$ in string graphs.

For the locally surjective homomorphism, defined analogously, the situation seems more complicated. We show the dichotomy theorem for simple connected graphs $H$ with maximum degree 2. If $H$ is isomorphic to $P_3$ or $C_4$, then the existence of a locally surjective homomorphism from a string graph with $n$ vertices to $H$ can be decided in time $2^{O(n^{2/3} \log^{3/2} n)}$, otherwise, assuming ETH, the problem cannot be solved in time $2^{o(n)}$.

As a byproduct, we obtain results concerning the complexity of variants of homomorphism problem in $P_t$-free graphs -- in particular, the weighted homomorphism dichotomy, analogous to the one for string graphs.
\end{abstract}
\end{titlepage}

\section{Introduction}
The theory of NP-completeness gives us tools to identify problems which are unlikely to admit polynomial-time algorithms, but it does not give any insight into possible complexities of problems that are considered hard.
For example, the best algorithms we know for most canonical problems like {\sc 3-Coloring}, {\sc Independent Set}, {\sc Dominating Set}, {\sc Vertex Cover}, {\sc Hamiltonian Cycle}, are single-exponential, i.e., with complexity $2^{O(n)}$ ($n$ will always denote the number of vertices in the input graph). On the other hand, in planar graphs these problems are still NP-complete, but they admit {\em subexponential} algorithms (i.e., working in time $2^{o(n)}$). Indeed, most canonical problems in planar graphs admit a certain ``square-root phenomenon'', i.e., can be solved in time $2^{\widetilde{O}(\sqrt{n})}$\footnote{in the $\widetilde{O}(\cdot)$ notation we suppress polylogarithmic factors}. The core building block in construction of subexponential algorithms for planar graphs is the celebrated planar decomposition theorem by Lipton and Tarjan \cite{LiptonTarjan}, which asserts that every planar graph has a balanced separator of size $O(\sqrt{n})$.

To argue whether those algorithms are asymptotically optimal and, in general, to prove meaningful lower bounds on the complexity of hard problems, we need a stronger assumption than ``P $\neq$ NP''. Such a stronger assumption, commonly used in complexity theory, is the Exponential Time Hypothesis (ETH) by Impagliazzo and Paturi \cite{ImpagliazzoPaturi}, which implies that {\sc 3-Sat} with $n$ variables cannot be solved in time $2^{o(n)}$. For example, assuming the ETH, {\sc 3-Coloring}, {\sc Independent Set}, {\sc Dominating Set}, {\sc Vertex Cover}, {\sc Hamiltonian Cycle} cannot be solved in time $2^{o(n)}$ in general graphs or in time $2^{o(\sqrt{n})}$ in planar graphs. Thus the algorithms we know are asymptotically tight, unless the ETH~fails.

A natural direction of research is to consider restricted graph classes and try to classify problems solvable in subexponential time. Geometric intersection graphs provide a rich family of graph classes, which are potentially interesting from the point of view of fine-grained complexity, as they lie ``in between'' planar graphs and all graphs.
For a family $\cS$ of sets, we define its {\em intersection graph}, whose vertices are in one-to-one correspondence to members of $\cS$, and two vertices are adjacent if and only if their corresponding sets intersect. We will be interested in intersection graphs of sets of geometric objects in the plane.

For example, in {\em unit disk graphs}, i.e., intersection graphs of unit-radius disks in the plane, {\sc Independent Set}, {\sc Hamiltonian Cycle}{, {\sc Vertex Cover} can be solved in time $2^{\tilde{O}(\sqrt{n})}$ \cite{DBLP:journals/jal/AlberF04,MarxP15,DBLP:conf/icalp/FominLPSZ17}, and {\sc $k$-Coloring} can be solved in time $2^{\tilde{O}(\sqrt{nk})}$ for every $k$ \cite{DBLP:conf/ciac/Kisfaludi-BakZ17,Biro17}.  All these bounds are essentially tight under the ETH, up to polylogarithmic factors in the exponent.
Many algorithms for (unit) disk graphs use the fact that disk intersection graphs also have small separators. Indeed, Miller {\em et al.} showed that the intersection graph of a family of $n$ disks, such that at most $k$ of them share a single point, has a balanced separator of size $O(\sqrt{nk})$ \cite{DBLP:journals/jacm/MillerTTV97}. This was later generalized to intersection graphs of families of arbitrary convex shapes that are fat, i.e., with bounded ratio of the radius of the smallest enclosing circle to the radius of the largest enclosed circle \cite{SmithW98}.

It is perhaps interesting to note that, by the celebrated kissing lemma by Koebe \cite{Koebe}, every planar graph is an intersection graph of interior-disjoint disks. Note that in such a representation each point is contained in at most two disks, so the separator theorem for disk graphs implies the planar separator theorem.

In this paper we are interested in intersection graphs of non-fat geometric objects.
In particular, we will investigate {\em string graphs}, i.e., intersection graphs of continuous curves in the plane (see Kratochv\'il \cite{DBLP:journals/jct/Kratochvil91,DBLP:journals/jct/Kratochvil91a}) and {\em segment graphs}, i.e., intersection graphs of straight-line segments (see Kratochv\'il and Matou\v{s}ek \cite{KRATOCHVIL1994289}). We can restrict the representation even further and consider \kdir{$k$} graphs, which are intersection graphs of segments using at most $k$ distinct slopes \cite{KRATOCHVIL1994289}. It is known that planar graphs form another subclass of segment graphs \cite{ChalopinG09,DBLP:conf/soda/GoncalvesIP18}.

General string separator theorems have been proven by Fox and Pach~\cite{DBLP:journals/cpc/FoxP10} and Matou\v{s}ek \cite{DBLP:journals/cpc/Matousek14}. The following, asymptotically tight version, was shown by Lee~\cite{Lee16}.

\begin{theorem}[Lee~\cite{Lee16}]\label{thm-stringsep}
Every string graph with $m$ edges has a balanced separator of size $O(\sqrt{m})$. It can be found in polynomial time, if the geometric representation is given.  \qed
\end{theorem}
Observe that since planar graphs are string graphs and have linear number of edges, \autoref{thm-stringsep} implies the planar separator theorem.

Using the string separator theorem, Fox and Pach~\cite{FoxP11} showed that {\sc Independent Set} (and thus {\sc Vertex Cover}) can be solved in subexponential time in string graphs. Combining \autoref{thm-stringsep} with their approach gives the complexity $2^{\tilde{O}(n^{2/3})}$.  The algorithm is a simple win-win strategy: either we have a vertex of large degree and we branch on choosing it to the solution or not, or all degrees are small and thus there exists a small balanced separator, which allows us for one step of divide \& conquer.
Recently, Marx and Pilipczuk \cite{MarxP15} used a different approach to obtain a $2^{{O}(\sqrt{n})} p^{O(1)}$ algorithm for {\sc Independent Set} in string graphs, where $p$ is the number of {\em geometric vertices} in the representation.

While the algorithm of Marx and Pilipczuk seems difficult to generalize to other problems, Bonnet and Rzążewski \cite{DBLP:conf/wg/BonnetR18} showed that the win-win strategy of Fox and Pach can be successfully applied to obtain subexponential algorithms for {\sc 3-Coloring}, {\sc Feedback Vertex Set}, and {\sc Max Induced Matching}. 
Quite surprisingly, they showed that for every $k\geq 4$, {\sc $k$-Coloring} cannot be solved in time $2^{o(n)}$, even in \kdir{2} graphs, unless the ETH fails.
They also showed that assuming the ETH, {\sc Dominating Set}, {\sc Independent Dominating Set}, and {\sc Connected Dominating Set} do not admit subexponential algorithms in segment graphs, and {\sc Clique} does not admit such an algorithm in string graphs.

This shows that the complexity landscape in string and segment graphs appears to be much more interesting than in planar graphs or intersection graphs of fat objects. In order to understand which problems can be solved in subexponential time, it would be especially desirable to obtain full dichotomy theorems for some natural families of problems, instead of proving ad-hoc results for single problems.
A natural language to describe these families in a uniform way is provided by graph homomorphisms. For graphs $G$ and $H$ (with possible loops), a {\em homomorphism} from $G$ to $H$, denoted by $h \colon G \to H$, is an edge-preserving mapping from the vertex set of $G$ to the vertex set of $H$ (see the book by Hell and Ne\v{s}et\v{r}il \cite{ Hell2004}).
A homomorphism $h \colon G \to H$ will be often called an $H$-{\em coloring} of $G$ and we will think of vertices of $H$ as {\em colors}.
Note that the notion of homomorphisms is flexible and allows us to impose additional restrictions, such as vertex/edge lists \cite{DBLP:journals/combinatorica/FederHH99,DBLP:conf/soda/HellR11} or vertex/edge weights \cite{DBLP:journals/ejc/GutinHRY08}.
This way many well-known problems can be formulated as problems of finding a homomorphism to a certain graph $H$, possibly with additional constraints. For example, {\sc $k$-Coloring} is equivalent to a homomorphism to $K_k$, and {\sc Independent Set} is equivalent to  a weight-maximizing homomorphism to $H=$
\begin{tikzpicture}[scale=.7, every node/.style={draw,circle,fill=white,inner sep=1pt,minimum size=5pt}]
\draw[line width=1] (0,0) -- (1,0);
\node at (0,0) {\tiny 1} edge [line width=1,in=60,out=120,loop] ();
\node at (1,0) {\tiny 0};
\end{tikzpicture}
(numbers denote weights of vertices of $G$ mapped to particular vertices of $H$).

\paragraph*{Weighted homomorphisms.} Let $H=(V(H),E(H))$ be a fixed graph (with possible loops), and consider the following computational problem called \whomo{$H$}. The instance consists of a graph $G=(V(G),E(G))$, a {\em weight function} $w \colon (V(G) \times V(H)) \cup (E(G) \times E(H)) \to \R$, and an integer $k$. For simplicity we also allow $-\infty$ as a weight, but this can be avoided by shifting all weights and using a sufficiently small integer to represent the weight corresponding to a forbidden choice.
For a homomorphism $h: G \to H$ and for any $X \subseteq V(G)\cup E(G)$ we define the weight of $X$ by $w_h(X)=\sum_{x \in X}w(x,h(x))$. The {\em weight of $h$} is defined as $w_h(V(G) \cup E(G))$. We ask if there exists a homomorphism from $G$ to $H$ whose total weight is at least $k$. It is straightforward to see that this problem generalizes some well-studied variants of graph homomorphism problem, including {\sc List Homomorphism} \cite{DBLP:journals/combinatorica/FederHH99} and {\sc Min Cost Homomorphism} \cite{DBLP:journals/ejc/GutinHRY08}.

We show the following dichotomy theorem for \whomo{$H$} in string graphs.

\begin{theorem} \label{thm-mainhom}
Let $H$ be a fixed graph.
\begin{compactenum}[(a)]
\item If $H$ has no two vertices with two common neighbors, then the \whomo{$H$} problem can be solved in time $2^{{O}(n^{2/3} \log n)}$ for string graphs with $n$ vertices.
\item Otherwise, the \whomo{$H$} problem in NP-complete and cannot be solved in time $2^{o(n)}$ for segment graphs, unless the ETH fails.
\end{compactenum}
\end{theorem}

Very recently Groenland {\em et al.} \cite{DBLP:journals/corr/abs-1803-05396} observed that if $H$ has no two vertices with two common neighbors, then  \whomo{$H$} can be solved in time $2^{O(\sqrt{n})}$ in $P_t$-free graphs, for every fixed $t$. Note that there are string graphs with arbitrarily long induced paths, and there are $P_t$-free graphs that are not string graphs.

The algorithm proving \autoref{thm-mainhom} (a) is a slight adaptation of the win-win approach by Fox and Pach \cite{FoxP11}, later extended by Bonnet and Rzążewski \cite{DBLP:conf/wg/BonnetR18}, and Groenland {\em et al.} \cite{DBLP:journals/corr/abs-1803-05396}.
The proof of part (b) is divided into a few cases, depending on the structure of $H$. In our reductions we try not to use the whole expressibility of the \whomo{$H$} problem, but aim to obtain hardness even for some natural special cases.
All hardness proofs follow the same pattern -- we start with a grid-like arrangement of segments, inducing a clique or a biclique, and then add constant-size gadgets to encode a specific problem. Note that this requires the objects to be non-fat and gives some intuition why problems in segment graphs tend to be harder than in intersection graphs of fat objects, and how the hardest instances look like. Finally, all graphs we construct are actually $P_t$-free for some fixed $t$, which, along with the result of Groenland {\em et al.} \cite{DBLP:journals/corr/abs-1803-05396}, gives the following dichotomy theorem for $P_t$-free graphs.

\begin{restatable}{theorem}{ptfreethm}\label{thm-mainptfree}
Let $H$ be a fixed graph.
\begin{compactenum}[(a)]
\item If $H$ has no two vertices with two common neighbors, then for all fixed $t$ the \whomo{$H$} problem can be solved in time $2^{{O}(\sqrt{n})}$ for $P_t$-free graphs with $n$ vertices.
\item Otherwise, the \whomo{$H$} problem is NP-complete and cannot be solved in time $2^{o(n)}$ for $P_t$-free graphs with $n$ vertices for some fixed $t$, unless the ETH fails.
\end{compactenum}
\end{restatable}

\paragraph*{Locally constrained homomorphisms.} Interesting variants of graph homomorphism problems can be obtained by imposing some additional constrains on the neighborhood of each vertex. A homomorphism $h$ from $G$ to $H$ is called {\em locally injective} ({\em locally bijective}, {\em locally surjective}) if for every $v \in V(G)$ it induces an injective (bijective, surjective, resp.) mapping between the neighborhood of $v$ and the neighborhood of $h(v)$. Locally bijective homomorphisms have been studied from combinatorial \cite{DBLP:conf/mfcs/FialaPT05,DBLP:journals/ejc/FialaM06} and computational point of view \cite{DBLP:journals/dm/MacGillivrayS10,DBLP:conf/wg/FialaK06,DBLP:journals/tcs/ChaplickFHPT15,DBLP:journals/tcs/FialaP05}. Let \lihomo{$H$}, \lbhomo{$H$}, and \lshomo{$H$} denote, respectively, the computational problems of determining the existence of a locally injective, bijective, and surjective homomorphism from a given graph to $H$.

Some well-known graph problems can be expressed as locally constrained homomorphism. For example, locally injective homomorphism to the complement of the $k$-vertex path appears to be equivalent to $k$-$L(2,1)$-labeling, i.e., a mapping from the vertex set of the input graph to the set $\{1,2,\ldots,k\}$, in which adjacent vertices get labels differing by at least 2, and vertices with a common neighbor get different labels \cite{DBLP:journals/siamdm/GriggsY92,DBLP:journals/algorithmica/HavetKKKL11}. If $H$ is the complete graph $K_k$, then \lihomo{$H$} is exactly the $k$-coloring of the square of the graph \cite{DBLP:journals/jgt/HeuvelM03,227702}. Finally, if $H$ is a complete graph with $k$ vertices, where every vertex has a loop, then \lihomo{$H$} is equivalent to the injective $k$-coloring \cite{DBLP:conf/latin/HellRS08,DBLP:journals/dm/HahnKSS02}, in which the only restriction is that no two vertices with a common neighbor get the same color.

We show that, unlike \whomo{$H$}, both \lihomo{$H$} and \lbhomo{$H$} can always be solved in subexponential time in string graphs.

\begin{restatable}{theorem}{mainlihom}\label{thm-mainlihom}
For every fixed graph $H$, the \lihomo{$H$} problem and the \lbhomo{$H$} problem can be solved in time $2^{{O}(\sqrt{n}\log n)}$ in string graphs.
\end{restatable}

The \lshomo{$H$} problem appears to be much harder.
In particular, we show the following dichotomy for simple graphs $H$ with $\Delta(H) \leq 2$ (observe that if $|H| \leq 2$, the problem can trivially be solved in polynomial time).

\begin{restatable}{theorem}{thmlsh} \label{thm:lsh-together}
Let $H$ be a connected simple graph with $\Delta(H) \leq 2$ and $|H| \geq 3$.
\begin{compactenum}[(a)]
\item If $H \in \{P_3,C_4\}$, then the \lshomo{$H$} problem can be solved in time $2^{O(n^{2/3}\log^{3/2} n)}$ for string graphs, even if geometric representation is not given.
\item Otherwise, the \lshomo{$H$} problem cannot be solved in time $2^{o(n)}$ in \kdir{$2$} graphs, unless the ETH fails.
\end{compactenum}
\end{restatable}

We also show that \lshomo{$H$} cannot be solved in subexponential time for $H=$ \begin{tikzpicture}[scale=.7, every node/.style={draw,circle,fill=white,inner sep=1pt,minimum size=5pt}]
\draw[line width=1] (0,0) -- (1,0);
\node at (0,0) {} edge [line width=1,in=60,out=120,loop] ();
\node at (1,0) {};
\end{tikzpicture}.
Note that none of the graphs $H$, for which we obtain negative results for \lshomo{$H$} problem, has two vertices with two common neighbors. Thus they are all ``easy'' cases of  \whomo{$H$}.

\paragraph{Representation and robust algorithms.} When dealing with geometric intersection graphs, we need to be careful, whether the input consist of the graph along with the representation, or just the graph (with a promise that a geometric representation exists). This distinction might be crucial, since finding a representation is often a computationally hard task.

Recognizing string and segment graphs was shown to be NP-hard by Kratochv\'il \cite{KRATOCHVIL1994289}, and Kratochv\'il and Matou\v{s}ek \cite{KratochvilM91}, respectively. However, for a very long time it was unclear whether these problems are in NP. This is because there are string graphs, whose every representation requires exponential number of crossing points \cite{KratochvilM91} and there are segment graphs, whose every representation requires points with double exponential coordinates \cite{KRATOCHVIL1994289,DBLP:journals/jct/McDiarmidM13}.
Finally, Schaefer, Sedgwick, and \v{S}tefankovi\v{c} showed that recognizing string graph is in NP \cite{DBLP:journals/jcss/SchaeferSS03}, while recognizing segment graph appears to be complete for the complexity class $\exists \R$ \cite{Schaefer2017,DBLP:journals/corr/Matousek14}. This is a strong indication that the problem might not be in NP.

For these reasons, it is desirable for an algorithm not to require explicit representation. Such algorithms are called {\em robust} -- they either compute a solution, or report that the input graph does not belong to the required class. All algorithms presented in the paper are robust, but can be made slightly faster, if the representation is given. On the other hand, all hardness results hold even if the graph is given along with the geometric representation.

\section{Weighted homomorphism problem}
In this section we prove \autoref{thm-mainhom}. Let us first discuss the notation and some preliminary observations.
For every vertex $v$ of graph $G$, let $N_G(v)$ denote the set of neighbors of $v$ in $G$ and for any $V' \subseteq V(G)$ let $N_G(V')=\bigcup_{v\in V'}N_G(v)$. Let $d_G(v)=|N_G(v)|$. If the graph is clear from the context, we will omit the subscript $G$ and simply write $N(v)$ and $d(v)$.
By $\Delta(G)$ we denote $\max_{v \in V(G)}d(v)$.


Let us recall the definition of \whomo{$H$}. The instance consists of a graph $G$, an integer $k$, and a \emph{weight function} $w: \left( V(G) \times V(H) \right) \cup \left( E(G) \times E(H) \right) \rightarrow \mathbb{R}$.
For a homomorphism $h: G \to H$ and for any $X \subseteq V(G)\cup E(G)$ we define the weight of $X$ by $w_h(X)=\sum_{x \in X}w(x,h(x))$. The {\em weight of $h$} is defined as $w_h(V(G) \cup E(G))$.
We ask if there exists a homomorphism, whose weight is at least~$k$.
A homomorphism $h \colon G \to H$ will be often called a {\em coloring} of $G$ and we will think of vertices of $H$ as {\em colors}.


We will say that $H$ (with possible loops) has \emph{\prop}, if it does not contain any pair of vertices $u,v$ such that $|N(u) \cap N(v)|\geq 2$. Note that if $H$ is loopless, then it has \prop if and only if it does not contain a copy of $C_4$ as a (non-necessarily induced) subgraph.


\subsection{Algorithm} \label{sec:algorithm}
In this section we prove statement (a) of \autoref{thm-mainhom}.
\begin{customtheorem}{2 (a)}
Let $H$ be a fixed graph. If $H$ has no two vertices with two common neighbors, then the \whomo{$H$} problem can be solved in time $2^{{O}(n^{2/3} \log n)}$ for string graphs with $n$ vertices.
\end{customtheorem}

\begin{proof}
Let $G$ be a graph and $w$ be a weight function. We will find a homomorphism from $G$ to $H$ of maximum total weight (if one exists). Actually, we will assume that we are additionally given lists $L \colon V(G) \to 2^{V(H)}$, and we ask for a homomorphism respecting these lists. Note that this is not really necessary, as list can be expressed with appropriate choice of weights, but this requires modifying the weight function and makes the argument more complicated.
Define $N:=\sum_{v\in V(G)}|L(v)|$. Let $\ell := 2^{|H|}$, note that there are at most $\ell$ possible lists $L(v)$.

We start with preprocessing the instance. For any two adjacent vertices $u$ and $v$ of $G$, if $L(v)$ contains a vertex $b \in V(H)$, which is non-adjacent to every vertex from $L(u)$, we can safely remove $b$ from $L(u)$. Moreover, if there exists a vertex $v \in V(G)$, such that $L(v)$ is a singleton, say $L(v)=\{a\}$, then we can map $v$ to $a$ and remove it from $G$. We repeat these steps while possible, this  clearly takes only polynomial time, as $N \leq n|H|$. 
Due to this step, we can assume that every list has at least two elements and for every $uv \in E(G)$ and every $a \in L(v)$ there exists $a' \in L(u)$ such that $aa' \in E(H)$.

First, consider the case that $G$ has a vertex $v$ such that $d(v)> n^{1/3}$. It means that there is a list $L$, which is assigned to at least $n^{1/3}/ \ell$ neighbors of $v$. We observe that there exist $a \in L(v)$ and $b \in L$ such that $ab \not\in E(H)$. Indeed, due to the preprocessing step, $|L(v)| \geq 2$ and $|L| \geq 2$. So if every element of $L(v)$ is adjacent to every element of $L$, then the \prop is violated.

We branch on assigning $a$ to $v$: either we remove $a$ from $L(v)$ or we color $v$ with $a$ and update the lists of neighbors of $v$ by removing from them every non-neighbor of $a$. In particular, we will remove $b$ from lists of at least $n^{1/3}/ \ell$ neighbors of $v$. Let $F(N)$ be the complexity of this step, clearly is it bounded by the following:
\begin{align*}
F(N) \leq &  F(N-1)+F ( N-n^{1/3}/ \ell) \leq  \left(n^{1/3}/ \ell +1\right)^{N \ell/ n^{1/3}} \\
\leq & 2^{O(N \log n / n^{1/3})} \leq  2^{{O}(n^{2/3}\log n)}.
\end{align*}
In the other case we have that every vertex of $G$ has degree at most $n^{1/3}$. This means that $G$ has $O(n^{4/3})$ edges, so, by Theorem~\ref{thm-stringsep}, there exists a balanced separator $S$ of size $O(n^{2/3})$. We can find $S$ in polynomial time using the geometric representation, or by exhaustive guessing in time $n^{O(n^{2/3})} = 2^{O(n^{2/3}\log n)}$, if the geometric representation is not given. Then, we consider all possible $H$-colorings of $S$ and run one step of a standard divide and conquer algorithm. The complexity of this step is $n^{O(n^{2/3})}= 2^{{O}(n^{2/3}\log n)}$ and so is the total combined complexity of the algorithm.
\end{proof}

Observe that even if the input graph $G$ is not a string graph, but has an appropriate structure of separators, the algorithm can still give the correct answer in time $2^{{O}(n^{2/3}\log n)}$. Otherwise, the exhaustive search for a separator will fail and we can report that $G$ is not a string graph. This means that the presented algorithm is robust.
Moreover, the algorithm can be easily adapted to count all feasible solutions.


\subsection{Hardness results} \label{sec:hard}
Now we show that \prop is essential for the existence of subexponential algorithms: for all remaining graphs $H$, an algorithm solving \whomo{$H$} for string graphs in subexponential time would contradict the ETH. To begin with, observe that we can express the \prop in terms of forbidden subgraphs.

\begin{observation}
A graph has \prop if and only if it does not contain any of the seven graphs shown on the Figure \ref{fig-themagnificentseven} as an induced subgraph. \qed
\end{observation}
\begin{figure}[h]
\begin{subfigure}[b]{0.13\textwidth}
\centering\begin{tikzpicture}[scale=.7, every node/.style={draw,circle,fill=white,inner sep=0pt,minimum size=5pt}]
\draw[line width=1] (0,0) -- (1,0);
\node[draw=none,fill=none] at (0.5,-0.5) {(a)};
\node at (0,0) {} edge [line width=1,in=60,out=120,loop] ();
\node at (1,0) {} edge [line width=1,in=60,out=120,loop] ();
\end{tikzpicture}
\end{subfigure}\hskip .1cm
\begin{subfigure}[b]{0.13\textwidth}
\centering\begin{tikzpicture}[scale=.7, every node/.style={draw,circle,fill=white,inner sep=0pt,minimum size=5pt}]
\draw[line width=1] (1,.5) -- (0,0);
\draw[line width=1] (1,.5) -- (0,1);
\draw[line width=1] (0,0) -- (0,1);
\node[draw=none,fill=none] at (0.5,-0.5) {(b)};
\node at (1,.5) {} edge [line width=1,in=-35,out=35,loop] ();
\node at (0,1) {} ;
\node at (0,0) {} ;
\end{tikzpicture}
\end{subfigure}\hskip .1cm
\begin{subfigure}[b]{0.13\textwidth}
\centering\begin{tikzpicture}[scale=.7, every node/.style={draw,circle,fill=white,inner sep=0pt,minimum size=5pt}]
\draw[line width=1] (0,1) -- (0,0);
\draw[line width=1] (0,1) -- (1,1);
\draw[line width=1] (0,0) -- (1,0);
\draw[line width=1] (1,1) -- (1,0);
\node[draw=none,fill=none] at (0.5,-0.5) {(c)};
\node at (0,0) {};
\node at (0,1) {};
\node at (1,0) {};
\node at (1,1) {};
\end{tikzpicture}
\end{subfigure}\hskip .1cm
\begin{subfigure}[b]{0.13\textwidth}
\centering\begin{tikzpicture}[scale=.7, every node/.style={draw,circle,fill=white,inner sep=0pt,minimum size=5pt}]
\draw[line width=1] (0,1) -- (0,0);
\draw[line width=1] (0,1) -- (1,1);
\draw[line width=1] (0,0) -- (1,0);
\draw[line width=1] (1,1) -- (1,0);
\node[draw=none,fill=none] at (0.5,-0.5) {(d)};
\node at (0,0) {};
\node at (0,1) {}  edge [line width=1,in=60,out=120,loop] ();
\node at (1,0) {};
\node at (1,1) {};
\end{tikzpicture}
\end{subfigure}\hskip .1cm
\begin{subfigure}[b]{0.13\textwidth}
\centering\begin{tikzpicture}[scale=.7, every node/.style={draw,circle,fill=white,inner sep=0pt,minimum size=5pt}]
\draw[line width=1] (0,1) -- (0,0);
\draw[line width=1] (0,1) -- (1,1);
\draw[line width=1] (0,0) -- (1,0);
\draw[line width=1] (1,1) -- (1,0);
\node[draw=none,fill=none] at (0.5,-0.5) {(e)};
\node at (0,0) {};
\node at (0,1) {} edge [line width=1,in=60,out=120,loop] ();
\node at (1,0) {} edge [line width=1,in=240,out=300,loop] ();
\node at (1,1) {};
\end{tikzpicture}
\end{subfigure}\hskip .1cm
\begin{subfigure}[b]{0.13\textwidth}
\centering\begin{tikzpicture}[scale=.7, every node/.style={draw,circle,fill=white,inner sep=0pt,minimum size=5pt}]
\draw[line width=1] (0,1) -- (0,0);
\draw[line width=1] (0,1) -- (1,1);
\draw[line width=1] (0,0) -- (1,0);
\draw[line width=1] (1,1) -- (1,0);
\draw[line width=1] (0,1) -- (1,0);
\node[draw=none,fill=none] at (0.5,-0.5) {(f)};
\node at (0,0) {};
\node at (0,1) {};
\node at (1,0) {};
\node at (1,1) {};
\end{tikzpicture}
\end{subfigure}\hskip .1cm
\begin{subfigure}[b]{0.13\textwidth}
\centering\begin{tikzpicture}[scale=.7, every node/.style={draw,circle,fill=white,inner sep=0pt,minimum size=5pt}]
\draw[line width=1] (0,1) -- (0,0);
\draw[line width=1] (0,1) -- (1,1);
\draw[line width=1] (0,0) -- (1,0);
\draw[line width=1] (1,1) -- (1,0);
\draw[line width=1] (0,1) -- (1,0);
\draw[line width=1] (0,0) -- (1,1);
\node[draw=none,fill=none] at (0.5,-0.5) {(g)};
\node at (0,0) {};
\node at (0,1) {};
\node at (1,0) {};
\node at (1,1) {};
\end{tikzpicture}
\end{subfigure}
\caption{Characterization of \prop by forbidden induced subgraphs.}
\label{fig-themagnificentseven}
\end{figure}
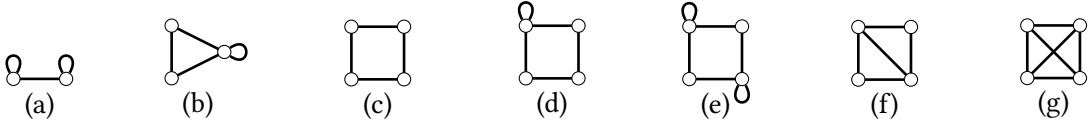

Note that to prove \autoref{thm-mainhom} (b), it is enough to show hardness of \whomo{$H$} for graphs $H$ shown in Fig. \ref{fig-themagnificentseven}.
Indeed, let $H$ be an induced subgraph of $H'$ and consider an instance $(G,w,k)$ of \whomo{$H$}. Define $w' \colon V(G) \times V(H') \cup E(G) \times E(H') \to \R$ as follows: for $x \in V(G) \cup E(G)$, if $a \in V(H) \cup E(H)$, then $w'(x,a)=w(x,a)$, otherwise $w'(x,a) = -\infty$. Note that $(G,w',k)$ is an instance of \whomo{$H'$}, equivalent to the instance $(G,w,k)$ of~\whomo{$H$}.

We prove Theorem \autoref{thm-mainhom} (b) for the graph (a) in Section \ref{sec-mc}, for (b) in Section \ref{sec-oct}, and for all the remaining cases in Section \ref{sec-c}.
Note that the problem  of finding a homomorphism to $K_4$ (the graph (g)) is exactly {\sc $4$-Coloring}. It is known that assuming the ETH, this problem does not admit a subexponential algorithm, even for \kdir{2} graphs~\cite{DBLP:conf/wg/BonnetR18}.

\subsubsection{Maximum Cut.}\label{sec-mc}
In this section, $H$ is the graph (a) from Fig. \ref{fig-themagnificentseven}. Note that any function $h \colon V(G) \to V(H)$ is a homomorphism and thus determining the existence of a homomorphism with just vertex lists and weights is trivial. However, it becomes more interesting if we include the edge weights.

We denote the vertices of $H$ by $a$ and $b$, so we have $E(H)=\{aa,ab,bb\}$ (see Fig.~\ref{fig-mch}). We also define the weight function as follows. Let $w(v,u)=0$ for every $(v,u) \in V(G) \times V(H)$, and for every $e \in E(G)$ we set $w(e, aa) = w(e,bb)=0$ and $w(e,ab)=1$. Note that the value of $w(e,f)$ does not depend on $e$, so $w$ is in fact an edge-weighting of $H$ (see Fig.~\ref{fig-mch}). 

\begin{figure}[h]
\centering\begin{tikzpicture}[scale=1, every node/.style={draw,circle,fill=white,inner sep=0pt,minimum size=5pt}]
\draw[line width=1] (0,0) -- node[above,draw=none] {1}  (2,0);
\node[label = {below:$a$}] at (0,0) {} edge [line width=1,in=60,out=120,loop, label = {[yshift = 0.4cm]:0}] ();
\node[label = {below:$b$}] at (2,0) {} edge [line width=1,in=60,out=120,loop, label = {[yshift = 0.4cm]:0}] ();
\end{tikzpicture}
\caption{Graph $H$ with its corresponding weights defined by $w$}
\label{fig-mch}
\end{figure}
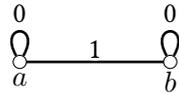

Observe that the weight of a homomorphism  $h \colon G \to H$ equals the number of edges mapped to $ab$. Thus finding a homomorphism of maximum weight is equivalent to partitioning the $V(G)$ into two subsets, so that the number of edges crossing this partition is maximized. Such a set of edges is called a {\em cut} in $G$ and the computational problem of finding the maximum cut is denoted by \maxcut. 

Thus for our result, it is enough to show the hardness of \maxcut on segment graphs.

\begin{restatable}{theorem}{mchomo} \label{thm-mchomo}
There is no algorithm solving \maxcut for segment graphs $G$ on $n$ vertices in the time $2^{o(n)}$, unless the ETH fails.
\end{restatable}

To prove Theorem \ref{thm-mchomo}, we present a sequence of linear reductions, starting from a well-known problem \pnaesat. In \pnaesat, for a given set of variables $U$ and clauses $C$, such that every clause contains at most 3 variables and all of them are non-negated, we ask if there exists a truth assignment $f:U \rightarrow \{0,1\}$ such that for every clause $c \in C$ we have that $f(c)=\{0,1\}$. It is known that, if we assume the ETH, there is no algorithm to solve \pnaesat on $n$ variables in time $2^{o(n)}$, even if each variable occurs at most 3 times. We start with the following lemma.
\begin{lemma}
There is no algorithm solving \maxcut in time $2^{o(n)}$ for graphs $G$ on $n$ vertices, even if $\Delta(G)\leq 6$, unless the ETH fails.
\label{lem-mc}
\end{lemma}
\begin{proof}
Let $\Phi=(U, C)$ be a \pnaesat formula such that $U=\{u_1,\dots, u_n\}$ and $C=\{c_1,\dots,c_m\}$. We can assume that all clauses are distinct and each clause $c_i$ contains 2 or 3 variables, as if $|c_i|=1$, then clearly the satisfying assignment does not exist. Let $m_2$ and $m_3$, respectively, be the number of clauses of each cardinality.

For $\Phi$, we construct a graph $G$, in which $\Delta(G)\leq 6$ and $|V(G)|=O(n)$. For every variable $u_i$ we create a {\em variable vertex} $x_i$.
For every 2-element clause $c = (u_i,u_j)$, we add an edge $x_ix_j$. For every 3-element clause $c = (u_i,u_j,u_k)$, we add a gadget containing six new vertices $l_1, r_1,l_2, r_2,l_3, r_3$ and nine new edges in a way that $(l_1,x_i, r_1,l_2,x_j, r_2,l_3, x_k,r_3)$ induce a 9-cycle (see Figure \ref{fig-clauses}~(a)). 

Let us now show that $\Phi$ is satisfiable if and only if $G$ has a cut of size at least $m_2+8m_3$.
First, assume that $\Phi$ is satisfiable. Let $f$ be a satisfying assignment and let us define the cut $(A,B)$ of $G$. For every $i$, we set $x_i \in A$ iff $f(u_i)=0$. Since $f$ is a satisfying assignment, all edges coming from 2-element clauses are in the cut. If $v$ is not a variable vertex, observe that it is adjacent to exactly one variable vertex $x_i$, and we set $v \in A$ iff $f(u_i)=1$. This way, since $f$ is satisfying, in each clause gadget we have eight edges in the cut, which gives $m_2+8m_3$ in total (see \autoref{fig-clauses} (b)).

Now let $(A,B)$ be a cut in $G$ of size at least $m_2+8m_3$. This means it has size exactly $m_2+8m_3$, as $|E(G)|=m_2+9m_3$ and there are $m_3$ edge-disjoint cycles $C_9$ in $G$. Moreover, exactly 8 edges from each clause gadget belong to the cut, so exactly two adjacent vertices are in the same part of cut. This gives us that for every clause at least one of its variable vertices is in  $A$ and at least one is in $B$ (see Figure \ref{fig-clauses}~(b)). Also, as the cut has size $m_2+8m_3$, all edges from 2-element clause gadgets must belong to the cut as well. So the truth assignment $f:U\rightarrow \{0,1\}$ such that $f(u_i)=0$ iff $x_i \in A$ satisfies  $\Phi$.

\begin{figure}[h]
\begin{subfigure}[b]{0.5\textwidth}
\centering\begin{tikzpicture}[scale=0.6,every node/.style={draw,circle,fill=white,inner sep=0pt,minimum size=5pt}]
\node[draw=none] at (-1,-1) {a)};
\node (a) at (0,1.5) {};
\node[draw=none,fill=none] at (-0.4,1.5) {\footnotesize{$x_i$}};
\node (b) at (3,1.5) {};
\node[draw=none,fill=none] at (2.6,1.6) {\footnotesize{$x_j$}};
\node (c) at (6,1.5) {};
\node[draw=none,fill=none] at (6.5,1.5) {\footnotesize{$x_k$}};
\node (la) at (0,-1.5) {};
\node[draw=none,fill=none] at (-0.3,-1.5) {\scriptsize{$l_1$}};
\node (ra) at (0.6,0) {};
\node[draw=none,fill=none] at (0.8,-0.4) {\scriptsize{$r_1$}};
\node (lb) at (2.2,0) {};
\node[draw=none,fill=none] at (2,-0.4) {\scriptsize{$l_2$}};
\node (rb) at (3.8,0) {};
\node[draw=none,fill=none] at (3.9,-0.4) {\scriptsize{$r_2$}};
\node (lc) at (5.4,0) {};
\node[draw=none,fill=none] at (5.1,-0.4) {\scriptsize{$l_3$}};
\node (rc) at (6,-1.5) {};
\node[draw=none,fill=none] at (6.5,-1.5) {\scriptsize{$r_3$}};
\draw (la) -- (a);
\draw (a) -- (ra);
\draw (ra) -- (lb);
\draw (lb) -- (b);
\draw (b) -- (rb);
\draw (rb) -- (lc);
\draw (lc) -- (c);
\draw (c) -- (rc);
\draw (rc) -- (la);
\end{tikzpicture}
\end{subfigure}%
\begin{subfigure}[b]{0.5\textwidth}
\centering\begin{tikzpicture}[scale=0.6,every node/.style={draw,circle,fill=white,inner sep=0pt,minimum size=5pt}]
\node[draw=none] at (-1,-1) {b)};
\node[fill=black] (a) at (0,1.5) {};
\node[draw=none,fill=none] at (-0.4,1.5) {\footnotesize{$x_i$}};
\node[fill=yellow] (b) at (3,1.5) {};
\node[draw=none,fill=none] at (2.6,1.6) {\footnotesize{$x_j$}};
\node[fill=yellow] (c) at (6,1.5) {};
\node[draw=none,fill=none] at (6.5,1.5) {\footnotesize{$x_k$}};
\node[fill=yellow] (la) at (0,-1.5) {};
\node[draw=none,fill=none] at (-0.3,-1.5) {\scriptsize{$l_1$}};
\node[fill=yellow] (ra) at (0.6,0) {};
\node[draw=none,fill=none] at (0.8,-0.4) {\scriptsize{$r_1$}};
\node[fill=black] (lb) at (2.2,0) {};
\node[draw=none,fill=none] at (2,-0.4) {\scriptsize{$l_2$}};
\node[fill=black] (rb) at (3.8,0) {};
\node[draw=none,fill=none] at (3.9,-0.4) {\scriptsize{$r_2$}};
\node[fill=black] (lc) at (5.4,0) {};
\node[draw=none,fill=none] at (5.1,-0.4) {\scriptsize{$l_3$}};
\node[fill=black] (rc) at (6,-1.5) {};
\node[draw=none,fill=none] at (6.5,-1.5) {\scriptsize{$r_3$}};
\draw (la) -- (a);
\draw (a) -- (ra);
\draw (ra) -- (lb);
\draw (lb) -- (b);
\draw (b) -- (rb);
\draw[style=dashed]  (rb) -- (lc);
\draw (lc) -- (c);
\draw (c) -- (rc);
\draw (rc) -- (la);
\end{tikzpicture}
\end{subfigure}
\caption{(a) A construction of gadget for a clause $c = (u_i,u_j,u_k)$ and (b) a maximum cut in this gadget.}
\label{fig-clauses}
\end{figure}
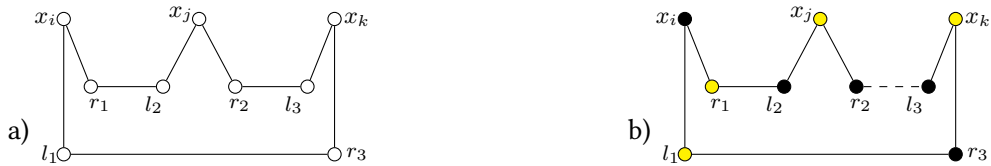
Finally, we observe that as each variable occurred in at most 3 clauses, we have $m\leq 3n$, so $|V(G)|\leq n+6m=O(n)$, which completes the proof.
\end{proof}

We say that a cut $(A,B)$ of $G$ is a \emph{bisection}, if $|A|=|B|$. We will use the following lemma, which says that \maxcut remains hard for graphs $G$ of bounded degree, even if every maximum cut in $G$ is a bisection. 

\begin{lemma} \label{lem:bisection}
Assuming the ETH, there is no algorithm solving \maxcut in time $2^{o(n)}$ for a graph $G$ on $n$ vertices, even if $\Delta(G)\leq 7$ and every maximum cut in $G$ is a bisection.
\end{lemma}
\begin{proof}
Let $(G,k)$ be an instance of \maxcut, such that $\Delta(G)\leq 6$ and $|V(G)|=n$. Construct a graph $F$ as follows: take a copy of $G$, denoted by $G'$, and for each its vertex $v$ add a vertex $v'$ and an edge $vv'$. Clearly, $|V(F)|=2n$ and $\Delta(F)\leq 7$.

First, we observe that if $(A_F,B_F)$ is a maximum cut of $F$, then it is a bisection. Assume the opposite, i.e., $(A_F,B_F)$ is a maximum cut such that $|A_F|<|B_F|$. As $|V(G')|=|V(F)\setminus V(G')|$, there exists $v \in V(G')$, such that $v, v' \in B_F$. Note that $(A_F \cup \{v'\}, B_F \setminus \{v'\})$ is a cut of larger size, a contradiction.

Now let us show that graph $G$ has a cut of size at least $k$ if and only if there exists a bisection in $F$ of size at least $n+k$. Assume that $(A,B)$ is a cut in $G$ of size at least $k$. Let $(A_F,B_F)$ be a cut in $F$ such that every vertex $v$ of $G'$ is in $A_F$ iff $v \in A$, and every vertex $v' \in V(F)\setminus V(G')$ is in $A_F$ iff its neighbor is in $B$.
Note that $(A_F,B_F)$ is a bisection, because $|A_F|=|B_F|=|A|+|B|$. There are at least $k$ edges between $A_F\cap V(G')$ and $B_F\cap V(G')$ and another $n$ edges between $V(G')$ and $V(F)\setminus V(G')$, so $(A_F,B_F)$ is a bisection of $F$ of size at least $n+k$.
For the converse, assume that $(A_F,B_F)$ is a bisection in $F$ of size at least $n+k$. Observe that at most $n$ edges in $F$ do not belong to $E(G')$, which means $(A_F \cap V(G'), B_F \cap V(G'))$ is a cut of $G$ of size at least $k$. 

This, combined with Lemma \ref{lem-mc}, completes the proof.
\end{proof}

Finally, we are ready to prove \autoref{thm-mchomo}.

\mchomo*

\begin{proof}
Let $(G,k)$ be a given instance of \maxbis, such that $\Delta(G)\leq 7$ and every maximum cut in $G$ is a bisection. Let $V(G)=\{v_1, \dots, v_n\}$ and $m:=|E(G)|$. We provide an instance $G^*$ of \maxcut, where $G^*$ is a segment graph, such that $G$ has a bisection of size at least $k$ if and only if $G^*$ has a cut of size at least \siz, which equivalently means that there exists a homomorphism from $G$ to $H$ of weight at least \siz.

We start constructing $G^*$ with two sets $X=\{x_1, \dots, x_n\}$ and $Y=\{y_1, \dots, y_n\}$ of segments. The segments of each set intersect in a single point, and all segments altogether are arranged in a grid-like manner (see Figure~\ref{fig-mcgrid}). 
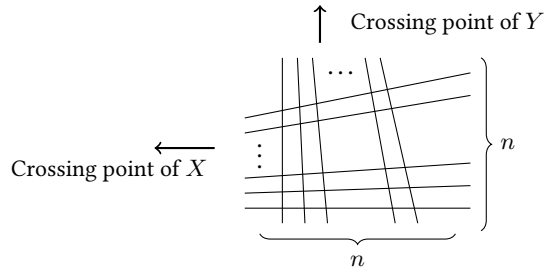
\begin{figure}[h]
\centering\begin{tikzpicture}[]
\draw (0,0) -- (3,0);
\draw (0,0.2) -- (3, 0.3); 
\draw (0,0.4) -- (3,0.6); 
\node at (0.2,0.8) {\vdots};
\draw (0,1) -- (3,1.5); 
\draw (0,1.2) -- (3,1.8); 
\draw (0.5,2) -- (0.5,-0.2);
\draw (0.7,2) -- (0.8, -0.2); 
\draw (0.9,2) -- (1.1,-0.2); 
\node at (1.3,1.8) {\ldots};
\draw (1.6,2) -- (2,-0.2); 
\draw (1.8,2) -- (2.3,-0.2);
\draw [decorate,decoration={brace,amplitude=4pt,mirror,raise=4pt},yshift=0pt]
(3,-0.3) -- (3,2) node [black,midway,xshift=0.5cm] {\footnotesize
{$n$}};
\draw [decorate,decoration={brace,amplitude=4pt,mirror,raise=4pt},yshift=0pt]
(0.2,-0.2) -- (2.8,-0.2) node [black,midway,yshift=-0.5cm] {\footnotesize
{$n$}};
\draw [thick,->] (-0.4,0.8) -- (-1.2,0.8) node [black,midway,yshift=-0.3cm,xshift=-1cm] {\footnotesize{Crossing point of $X$}};
\draw [thick,->] (1,2.2) -- (1,2.7) node [black,midway,xshift=1.7cm] {\footnotesize{Crossing point of $Y$}};
\end{tikzpicture}
\caption{First step of the construction of $G^*$.}
\label{fig-mcgrid}
\end{figure}

For each vertex $v_i$ of $G$, we add a {\em vertex gadget} on the intersection of $x_i$ and $y_i$. It contains a set $D_i$ of 16 parallel non-overlapping segments, crossing both $x_i$ and $y_i$ (see Figure~\ref{fig-mcgadgets}). For each edge $v_iv_j$ of $G$ we define an {\em edge gadget} on the intersection of $x_i$ and $y_j$ by putting two additional segments $\alpha_{ij}$ and $\beta_{ij}$, crossing each other and, respectively, $x_i$ or $y_j$. Note that each edge of $v_iv_j$ of $G$ is represented by two edge gadgets, one on the intersection of $x_i$ and $y_j$, and another one on the intersection of $x_j$ and $y_i$. Define $E_i:=\bigcup_{v_j \in N(v_i)}\{\alpha_{ij},\beta_{ij}\}$.
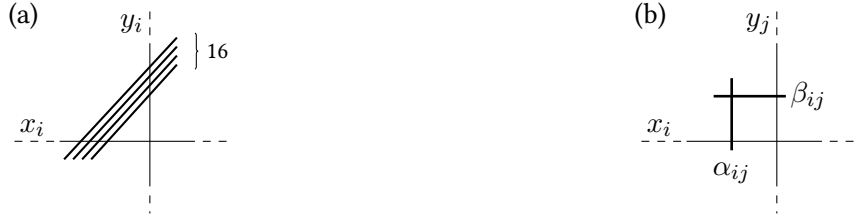
\begin{figure}[h]
\begin{subfigure}[b]{0.5\textwidth}
\centering\begin{tikzpicture}[scale=1.2]
\node at (0.1,2.9) {(a)};
\draw (1.5,1) -- (1.5,2.5);
\draw (0.5,1.5) -- (2, 1.5); 
\draw[line width=0.8] (0.85,1.3) -- (1.8,2.35); 
\draw[line width=0.8] (0.75,1.3) -- (1.8,2.45); 
\draw[line width=0.8] (0.65,1.3) -- (1.8,2.55); 
\draw[line width=0.8] (0.55,1.3) -- (1.8,2.65); 
\draw [decorate,decoration={brace,amplitude=1pt,mirror,raise=0pt}]
(2,2.3) -- (2,2.7) node [black,midway,xshift=0.3cm] {\footnotesize
{16}};
\draw[dashed] (1.5,0.9) -- (1.5,0.7);
\draw[dashed] (1.5,2.5) -- (1.5,3);
\draw[dashed] (0,1.5) -- (0.5,1.5);
\draw[dashed] (2.35,1.5) -- (2.05,1.5);
\node[fill=none, draw=none] at (.2, 1.65) {$x_i$};
\node[fill=none, draw=none] at (1.3,2.8) {$y_i$};
\end{tikzpicture}
\end{subfigure}%
\begin{subfigure}[b]{0.5\textwidth}
\centering\begin{tikzpicture}[scale=1.2]
\node at (0.1,2.9) {(b)};
\draw (1.5,1) -- (1.5,2.5);
\draw (0.5,1.5) -- (2, 1.5); 
\draw[line width=1](0.8,2) -- (1.6,2);
\draw[line width=1](1,2.2) -- (1,1.4);
\draw[dashed] (1.5,0.9) -- (1.5,0.7);
\draw[dashed] (1.5,2.5) -- (1.5,3);
\draw[dashed] (0,1.5) -- (0.5,1.5);
\draw[dashed] (2.35,1.5) -- (2.05,1.5);
\node[fill=none, draw=none] at (.2, 1.65) {$x_i$};
\node[fill=none, draw=none] at (1.3,2.8) {$y_j$};
\node[fill=none, draw=none] at (1,1.15) {$\alpha_{ij}$};
\node[fill=none, draw=none] at (1.85,2) {$\beta_{ij}$};
\end{tikzpicture}
\end{subfigure}
\caption{(a) A vertex gadget for $v_i \in V(G)$ and (b) an edge gadget for $v_iv_j \in E(G)$.}
\label{fig-mcgadgets}
\end{figure}

We say a homomorphism $h: G^* \rightarrow H$ is \emph{pretty}, if the following properties are satisfied:
\begin{compactitem}
\item[P1.] $|h(D_i)|=1$ for every $v_i \in V(G)$,
\item[P2.] $h(x_i)=h(y_i)$ and $h(x_i)\not\in h(D_i)$ for every $v_i \in V(G)$,
\item[P3.] $h(\alpha_{ij})\neq h(x_i)$ and $h(\beta_{ij})\neq h(y_j)$ for every $v_iv_j \in E(G)$.
\end{compactitem}

\begin{claimm}
If there exist a homomorphism $h:G^* \rightarrow H$ of weight $p$, then there exists a pretty homomorphism from $G^*$ to $H$ of weight at least $p$.
\label{cla-mcpretty}
\end{claimm}
\begin{inproof}
Consider a homomorphism $h \colon G^* \to H$. 
Note that recoloring all segments in $D_i$ to the color other than $h(x_i)$ does not decrease the weight, so we can assume that P1 holds for $h$.
So suppose that P2 does not hold and let $v_i$ be a vertex for which $h(x_i) \neq h(y_i)$, without loss of generality $h(x_i)=a$. Let $A$ ($B$, resp.) be the set of segments from $(X \cup Y) \setminus \{x_i,y_i\}$, that are mapped to $a$ ($b$, resp.) by $h$.

\noindent\textbf{Case 1:} If $|A|\geq |B|$, then consider a homomorphism $h':G^* \rightarrow H$, obtained from $h$ by recoloring $x_i$ to $b$ and all segments in $D_i$ to $a$.
We will show that $w_{h'}(E(G^*)) > w_h(E(G^*))$. Clearly, the weight can differ only on edges which have at least one endpoint in the set $D_i \cup \{x_i,y_i\}$, denote the set of these edges by $E'$. First, let us count the edges from $E'$ that were mapped to $ab$ by $h$, i.e., that contribute to $w_h(E(G^*))$.
There are $|A|+|B|+1$ such edges with both endpoints in $X \cup Y$, $16$ with one endpoint in $D_i$ and at most 14 with one endpoint in $E_i$, which gives at most $|A|+|B|+31$ in total.
Now let us count the edges from $E'$ that are mapped to $ab$ by $h'$. There are $2|A|$ such edges with both endpoints in $X \cup Y$, $32$ with one endpoint in $D_i$ and some number (possibly zero) of edges with one endpoint in $E_i$. Observe that $w_{h'}(E(G^*))-w_h(E(G^*)) \geq (2|A|+32) - (|A| + |B| + 31) > 0.$

\noindent\textbf{Case 2:} If $|A| <|B|$, then we consider a homomorphism $h':G^* \rightarrow H$, obtained from $h$ by recoloring $y_i$ to $a$ and all segments in $D_i$ to $b$. With the reasoning analogous to the one in Case 1, we show that the weight of $h'$ is at least the weight of $h$.

After at most $n$ such recolorings, we receive a homomorphism satisfying P1 and P2, so now we will assume that P1 and P2 hold in $h$. Suppose that there is $v_iv_j$ violating P3. Again, we recolor some segments to obtain a new homomorphism $h'$. 
If $h(x_i)=h(y_j)$, then we recolor $\alpha_{ij}$ and $\beta_{ij}$ to the color other than $h(y_i)$.
If $h(x_i)\neq h(y_j)$, then we color $\alpha_{ij}$ with the color $h(y_j)$, and $\beta_{ij}$ with the color $h(x_i)$. Again, observe that the weight of $h'$ is not lower than the weight of $h$ and $h'$ still satisfies P1 and P2.
Repeat this for every edge which does not satisfy P3. This way, after at most $2m$ steps, we obtain a pretty homomorphism.
\end{inproof}

Now let us show that $G$ has a bisection of size at least $k$ if and only if there exists homomorphism from $G^*$ to $H$ of weight at least \siz, i.e., a cut in $G^*$ of size at least \siz.

First assume there exists a bisection $(A_G,B_G)$ of $G$ of size at least $k$. Let us define a homomorphism $h: G^* \rightarrow H$. For every $v_i \in A_G$, we set $h(x_i)=h(y_i)=a$, and for every $v_i \in B_G$, we set $h(x_i)=h(y_i)=b$. Moreover, for every $i$, we color all segments from $D_i$ with the color other than $h(x_i)$. Finally, note that every $\alpha_{ij}$ and every $\beta_{ij}$ has exactly one neighbor in $X \cup Y$, and we color it with the color other than the color of this neighbor.

Let us count edges mapped to $ab$. Clearly, there are $n^2$ such edges with both endpoints in $X \cup Y$, and another $32n$ in vertex gadgets. For every $i,j$, such that $v_iv_j$ is in the cut $(A_G,B_G)$, we get another three edges mapped to $ab$ from the edge gadget, which gives $6k$ in total (recall that each edge is represented twice). If $v_iv_j$ is not in the cut, then only two edges from each edge gadget are mapped to $ab$, which gives $4(m-k)$ edges in total. Summing up, the weight of $h$ is \siz.

Now let $h: G^* \rightarrow H$ be a homomorphism of weight at least \siz. According to Claim \ref{cla-mcpretty}, we can assume that $h$ is pretty. Define a partition $(A_G,B_G)$ of $V(G)$ in the following way: if $h(x_i)=a$, then $v_i \in A_G$ and if $h(x_i)=b$, then $v_i \in B_G$. Let $k'$ be the number of edges between the sets $A_G$ and $B_G$.

As $h$ is pretty, we can easily count the edges mapped to $ab$. There are at most $n^2$ such edges between the vertices of $X \cup Y$, $32n$ edges with one endpoint in $\bigcup_{i \in [n]} D_i$, and $4m+2k'$ edges with at least one endpoint in $\bigcup_{i \in [n]} E_i$. This means that $h$ has weight of at most $n^2+32n+4m+2k'$. As we assumed that the weight is at least \siz, this implies that $k' \geq k$. This means that there exists a cut in $G$ of size at least $k$. As we know that each maximum cut of $G$ is a bisection, we know that there also exists a bisection of $G$ of size at least $k$.

To complete the proof, we observe that $G^*$ has $18n+4m \leq 30n=O(n)$ vertices. 
\end{proof}
\subsubsection{Minimum odd-cycle transversal.}\label{sec-oct}
Now let us consider the case when $H$ is the graph (b) in \autoref{fig-themagnificentseven}. This time we will consider a vertex-weighted variant. Denote the vertices of $H$ by $a,b,c$, where $a$ is the vertex with the loop.
All edge-weights are set to 0. For vertex weights, for every $v \in V(G)$ we set $w(v,b)=w(v,c)=1$ and $w(v,a)=0$. Note that again the weights do not depend on the choice of $v$, so we can think of $w$ as a vertex-weighting of $H$ (see Figure \ref{fig-oct}).

\begin{figure}[h]
\centering\begin{tikzpicture}[scale=1, every node/.style={draw,circle,fill=white,inner sep=1pt,minimum size=5pt}]
\draw[line width=1] (1,.5) -- (0,0);
\draw[line width=1] (1,.5) -- (0,1);
\draw[line width=1] (0,0) -- (0,1);
\node[label = {above:a}] at (1,.5) {0} edge [line width=1,in=-35,out=35,loop] ();
\node[label = {left:b}] at (0,1) {1};
\node[label = {left:c}] at (0,0) {1};
\end{tikzpicture}
\caption{Graph $H$ with its corresponding weights defined by $w$}
\label{fig-oct}
\end{figure}
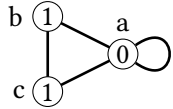

We observe that finding a homomorphism of maximum weight is equivalent to the problem of finding the maximum number of vertices of $G$ which induce a bipartite subgraph, or, equivalently, the minimum number of vertices, whose removal destroys all odd cycles. This problem is known as \oct. We will show the following.

\begin{theorem} \label{thm-octhomo}
The \oct problem in \kdir{$2$} graphs with $n$ vertices cannot be solved in time $2^{o(n)}$, unless the ETH fails.
\end{theorem}

\begin{proof}
This time we will reduce from the \is problem, which cannot be solved in time $2^{o(n)}$, even if the input graph has a bounded maximum degree.
Let $(G,k)$ be an instance of \is, such that $V(G)=\{v_1, \dots, v_n\}$, $|E(G)|=m = O(n)$.
We will construct a segment graph $G^*$, which admits a homomorphism to $H$ of weight at least $7n+2m+k$ if and only if $G$ has an independent set of size at least $k$.

We start the construction of $G^*$ with introducing the set $X=\{x_1, \dots, x_n\}$ of disjoint, parallel, horizontal segments and the set $Y=\{y_1, \dots, y_n\}$ of vertical segments, such that $X$ and $Y$ form a grid.
For every vertex $v_i$, on the intersection of $x_i$ and $y_i$ we add a {\em vertex gadget} shown on Figure \ref{fig-octgadgets} (a).  Let $D_i=\{d_1, \dots, d_7\}$.
For every $i,j$, such that $v_iv_j \in E(G)$, on the intersection of $x_i$ and $y_j$ we add an {\em edge gadget}  shown on Figure \ref{fig-octgadgets} (b). Note that again each edge is represented by two edge gadgets.
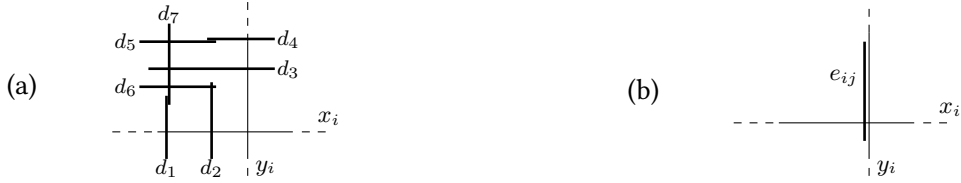
\begin{figure}[h]
\begin{subfigure}[b]{0.5\textwidth}
\centering\begin{tikzpicture}[scale=1.2]
\draw (1.5,1.2) -- (1.5,2.5);
\draw (0.5,1.5) -- (2, 1.5); 
\draw[line width=1] (1.1,2.05) -- (1.1, 1.2); 
\draw[line width=1] (.6,1.9) -- (.6, 1.2); 
\draw[line width=1] (.63,1.8) -- (.63, 2.7);
\draw[line width=1] (.3,2) -- (1.15,2); 
\draw[line width=1] (1.8,2.2) -- (.4, 2.2);
\draw[line width=1] (1.15,2.5) -- (.3, 2.5); 
\draw[line width=1] (1.05,2.53) -- (1.8, 2.53);
\draw[dashed] (1.5,1.1) -- (1.5,0.9);
\draw[dashed] (1.5,2.5) -- (1.5,3);
\draw[dashed] (0,1.5) -- (0.5,1.5);
\draw[dashed] (2.1,1.5) -- (2.4,1.5);
\node[fill=none, draw=none] at (2.4, 1.65) {\footnotesize{$x_i$}};
\node[fill=none, draw=none] at (1.7,1.1) {\footnotesize{$y_i$}};
\node[fill=none, draw=none] at (.6,1.1) {\scriptsize{$d_1$}};
\node[fill=none, draw=none] at (1.1,1.1) {\scriptsize{$d_2$}};
\node[fill=none, draw=none] at (1.95,2.2) {\scriptsize{$d_3$}};
\node[fill=none, draw=none] at (1.95,2.53) {\scriptsize{$d_4$}};
\node[fill=none, draw=none] at (0.15,2.5) {\scriptsize{$d_5$}};
\node[fill=none, draw=none] at (0.15,2) {\scriptsize{$d_6$}};
\node[fill=none, draw=none] at (.63,2.8) {\scriptsize{$d_7$}};
\node[fill=none, draw=none] at (-1,2) {(a)};
\end{tikzpicture}
\end{subfigure}%
\begin{subfigure}[b]{0.5\textwidth}
\centering\begin{tikzpicture}[scale=1.2]
\draw (1.5,1.1) -- (1.5,2.5);
\draw (0.5,1.5) -- (2, 1.5); 
\draw[line width=1] (1.45,1.3) -- (1.45,2.4); 
\draw[dashed] (1.5,1) -- (1.5,0.8);
\draw[dashed] (1.5,2.5) -- (1.5,2.8);
\draw[dashed] (0,1.5) -- (0.5,1.5);
\draw[dashed] (2.1,1.5) -- (2.4,1.5);
\node[fill=none, draw=none] at (2.4, 1.65) {\footnotesize{$x_i$}};
\node[fill=none, draw=none] at (1.7,1) {\footnotesize{$y_i$}};
\node[fill=none, draw=none] at (1.22,2) {\footnotesize{$e_{ij}$}};
\node[fill=none, draw=none] at (-1,1.85) {(b)};
\end{tikzpicture}
\end{subfigure}
\caption{(a) A vertex gadget for $v_i \in V(G)$ and (b) an edge gadget for $v_iv_j \in E(G)$.} 
\label{fig-octgadgets}
\end{figure}

We say that a homomorphism $h: G^* \rightarrow H$ is \emph{pretty}, if
\begin{compactenum}[P1.]
\item $h(e_{ij})\in\{b,c\}$ for every $v_iv_j \in E(G)$,
\item $h(x_i)=a$ iff $h(y_i)=a$ for every $v_i \in V(G)$.
\item if $h(x_i)=h(y_i)=a$, then $w_h(D_i)=7$, otherwise $w_h(D_i)=6$.
\end{compactenum}

\begin{claimm} \label{cla-octpretty}
Let $p >7n+2m$. If there exist a homomorphism $h:G^* \rightarrow H$ of weight $p$, then there exists a pretty homomorphism from $G^*$ to $H$ of weight at least $p$.
\end{claimm}
\begin{inproof}
First, observe that it is impossible that $\{b,c\} \subseteq h(X)$. Indeed, in such a case we would have $h(Y) = \{a\}$, and for every $v_i \in V(G)$ such that $h(x_i) \in \{b,c\}$, the total weight of $D_i$ is at most 6. Since the total weight of all segments in edge gadgets is at most $2m$, we obtain that $p\leq n + 6n+2m$, a contradiction.
Analogously we can show that $\{b,c\} \not\subseteq h(Y)$. Thus, by symmetry, we may assume that $h(X) \subseteq \{a,b\}$ and $h(Y) \subseteq \{a,c\}$.

Assume that P1 is not satisfied. Let $E'$ be the union of sets $\{e_{ij}, x_i, x_j, y_i, y_j\}$ over all $i,j$ violating P1, i.e., for which $e_{ij}=a$.
We will show that we can recolor some vertices from $E'$ in order to obtain a homomorphism $h'$ with weight at least $p$, satisfying P1. We will use the iterative procedure described below; if the color of some vertex is not specified, it means that it is not changed.

\noindent\textbf{Step 1.} For any $e_{ij}\in E'$ such that $h(x_i)=a$ or $h(y_j)=a$ set $h'(e_{ij})=b$ or $h'(e_{ij})=c$, respectively. Observe that $h'$ is a homomorphism and clearly $w_{h'}(\{e_{ij}\})>w_h(\{e_{ij}\})$, so the total weight is not decreased.

\noindent\textbf{Step 2.} For any $e_{ij} \in E'$ such that $h(e_{ji}) \in \{b,c\}$ observe that $h(x_j)=a$ or $h(y_i)=a$. If $h(x_j)=a$, then set $h'(y_j)=a$ and $h'(e_{ij})=c$. Otherwise, set $h'(x_i)=a$ and $h'(e_{ij})=b$. After that, go back to the Step 1. Note that $h'$ is a homomorphism and $w_{h'}(\{y_j,e_{ij},x_i,e_{ij}\})= w_h(\{y_j,e_{ij},x_i,e_{ij}\})$.

\noindent\textbf{Step 3.} If there is still some $e_{ij} \in E'$, observe that $h(x_i)=h(x_j)=b$ and $h(y_i)=h(y_j)=c$ and $h(e_{ij})=h(e_{ji})=a$. Set $h'(x_i)=h'(y_i)=a$ and $h'(e_{ij})=h'(e_{ji})=b$, then go back to the Step 1. We note that $h'$ is a homomorphism and $w_{h'}(\{x_i,y_i,e_{ij},e_{ji}\})= w_h(\{x_i,y_i,e_{ij},e_{ji}\})$.

We may change colors of vertices from $(X \cup Y) \cap E'$ more than once, but we always remove $e_{ij}$ from $E'$ and thus the procedure terminates after at most $2m$ steps. Thus from now on we may assume that P1 holds for $h$.

Now assume that P2 does not hold, without loss of generality suppose that $h(x_i)=b$ and $h(y_i)=a$. Observe that $w_h(D_i)\leq 6$. Let $h'$ be obtained from $h$ by recoloring the segment $x_i$ to the color $a$, the segments $d_1, d_3, d_5, d_6$ to the color $b$, and the segments $d_2, d_4, d_7$ to the color $c$. Clearly, $h'$ is a homomorphism and its weight is at least the weight of $h$. Moreover, if P1 holds for $h$, then it holds in $h'$ as well. After at most $n$ such recolorings we obtain a homomorphism satisfying P1 and P2.

Finally, suppose that P3 is violated for some $i$. If $h(x_i)=h(y_i)=a$, we can safely color $d_1, d_3, d_5, d_6$ to $b$ and $d_2, d_4, d_7$ to $c$. If $h(x_i)=b$ and $h(y_i)=c$, then note that at least one segment from the gadget must be colored to $a$. We can color $d_7$ with $a$ and other segments from $D_i$ with $b$ and $c$. Note that this does not decrease the weight. We repeat this step while possible, and after at most $n$ steps we obtain a pretty homomorphism.
\end{inproof}

Now let us show that $G$ has an independent set of size at least $k$ if and only if there exists a homomorphism $G^* \rightarrow H$ of weight at least $7n+2m+k$.
First, let $I$ be an independent set in $G$, such that $|I|\geq k$.
Define a mapping $h \colon V(G^*) \to \{a,b,c\}$ as follows. For every $v_i \in I$, we color $d_7$ to $a$ and $x_i,d_3,d_4,d_6$ to $b$ and $y_i,d_1,d_2,d_5$ to $c$. For every $v_i \notin I$ we set $x_i,y_i$ to $a$ and $d_1,d_3,d_5,d_6$ to $b$ and $d_2,d_4,d_7$ to $c$. Note that for every edge $v_iv_j$ we either have $h(x_i)=a$ or $h(y_j)=a$, so we can always color $e_{ij}$ with the color $b$ or $c$. Observe that the weight of such defined $h$ is at least $2k + 6k + 7(n-k) + 2m= 7n+2m+k$.

Now assume that there exists a homomorphism $h: G^* \rightarrow H$ of weight at least $7n+2m+k$. By Claim \ref{cla-octpretty} we can assume that $h$ is pretty. Let $I$ be the set of vertices $v_i$ of $G$, such that $h(x_i)=b$ and $h(y_i)=c$. If $I$ is not independent, then there exist $v_iv_j \in E(G)$, such that $h(x_i)=h(x_j)=b$ and $h(y_i)=h(y_j)=c$. But then observe it implies that $h(e_{ij})=h(e_{ji})=a$, which contradicts the fact that $h$ is pretty.
Let $k'$ be the size of $I$ and compute the weight of $h$. Since $h$ is pretty, every $e_{ij}$ contributes to the total weight. Moreover, $2k'$ vertices from $(X \cup Y)$ contribute to the weight. Finally, for every $v_i \in I$ we have $w_h(D_i) = 6$ and for every $v_i \notin I$ we have $w_h(D_i) = 7$. Summing up, the total weight of $h$ is $2m + 2k' + 6k' + 7(n-k') = 7n+2m+k'$. Since this is at least $7n+2m+k$, we conclude that $k' \geq k$ and thus $I$ is an independent set with at least $k$ vertices.

Finally, we observe that $|V(G^*)|\leq 2n+7n+2m=O(n)$, which completes the proof.
\end{proof}
\subsubsection{Other cases.}\label{sec-c}
Now we show the hardness result, when the target graph is $C_4$, i.e., the graph (c) in \autoref{fig-themagnificentseven}. 
This is actually the only place we use (almost) full expressibility of the \whomo{$H$} problem, i.e., the fact that edge weights may vary for different vertices of $G$. The graph created in the reduction will be a complete bipartite graph and most information will be encoded by the edge weights. Note that this is quite similar to the observation by Bodlaender and Jansen that \textsc{Weighted} \maxcut is NP-complete even for complete graphs~\cite{DBLP:journals/njc/BodlaenderJ00}.

\begin{theorem} \label{thm-homo-c4}
The \whomo{$C_4$} problem in \kdir{2} graphs with $n$ vertices cannot be solved in the time $2^{o(n)}$, unless the ETH fails.
\end{theorem}

\begin{proof}
Again, we will show a reduction from \is. The proof is analogous to the proof of \autoref{thm-octhomo}, but simpler. However, this time we will make use of edge weights.

Let $a,b,c,d$ be consecutive vertices of $C_4$. We will build an instance $G^*$ of \whomo{$C_4$}.
To define vertex weights, we set $w(v,a)=w(v,b)=1$ and $w(v,c)=w(v,d)=0$ for every $v \in V(G^*)$ (see \autoref{fig-ch}).
\begin{figure}[h]
\centering\begin{tikzpicture}[scale=.7, every node/.style={draw,circle,fill=white,inner sep=0pt,minimum size=9pt}]
\draw[line width=1] (0,0) -- (2,0);
\draw[line width=1] (0,2) -- (0,0);
\draw[line width=1] (2,2) -- (2,0);
\draw[line width=1] (0,2) -- (2,2);
\node at (0,0) {\scriptsize{1}};
\node at (2,0) {\scriptsize{1}};
\node at (0,2) {\scriptsize{0}};
\node at (2,2) {\scriptsize{0}};
\node[draw=none,fill=none] at (-0.4,0) {\footnotesize{a}};
\node[draw=none,fill=none] at (-0.4,2) {\footnotesize{d}};
\node[draw=none,fill=none] at (2.4,0) {\footnotesize{b}};
\node[draw=none,fill=none] at (2.4,2) {\footnotesize{c}};
\end{tikzpicture}
\caption{Graph $C_4$ with its corresponding vertex weights defined by $w$.}
\label{fig-ch}
\end{figure}
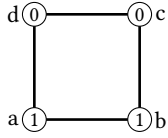

Let $(G,k)$ be an instance of \is with $V(G)=\{v_1, \dots, v_n\}$ and $|E(G)|=m=O(n)$.
Exactly like in section \ref{sec-oct}, we build a grid from two sets $X=\{x_1, \dots, x_n\}$ and $Y=\{y_1, \dots, y_n\}$, each containing $n$ disjoint parallel segments. There are no additional gadgets this time, only edge weights defined as follows. For every $i \in [n]$, we set $w(x_iy_i,e) = 0$ if $e \in \{ab,cd\}$ and $w(x_iy_i,e)=-\infty$ otherwise.
For every $i,j$, such that $v_iv_j \in E(G)$, we set $w(x_iy_j,e) = 0$ if $e \in \{cb,cd,ad\}$ and $w(x_iy_j,ab)=-\infty$.
All remaining edge weights are set to 0.

We claim that $G$ has an independent set of size at least $k$ if and only if there exists a weighted list homomorphism $h:G^*\rightarrow C_4$ of weight at least $2k$.  Note that in order to obtain a homomorphism with positive weight, no edge can obtain the weight $-\infty$. So edge weights are actually edge lists.
First, suppose there exists an independent set $I$ of size at least $k$ in $G$. We define a mapping $h$ as follows. For every $i$, such that $v_i \in I$, we set $h(x_i)=a$ and $h(y_i)=b$. For every $i$, such that $v_i \notin I$, we set $h(x_i)=c$ and $h(y_i)=d$. Note that $h$ is a homomorphism from $G^*$ to $C_4$. Moreover, if $v_i \in I$, then $h(x_iy_i)=ab$ and if $v_i \notin I$, then $h(x_iy_i)=cd$, so always $w_h(x_iy_i) \neq -\infty$. Also, note that for every edge $x_iy_j$ of $G$ we have $h(x_iy_j) \neq ab$, because $I$ is an independent set. This means that no edge gets weight $-\infty$ in $h$. Clearly the weight of $h$ is at least $2k$.

Now assume that there exists a homomorphism $h: G^* \rightarrow C_4$ with weight at least $2k$, respecting the lists. Observe that $G^*$ is bipartite, so one of its bipartition classes must be mapped to $\{a,c\}$ and another one to $\{b,d\}$. Without loss of generality assume that $h(X) \subseteq \{a,c\}$. Define a subset $I$ of the vertices of $G$ as follows: $v_i \in I$ if and only if $h(x_i)=a$. Let $k':=|I|$. We need to show that $I$ is an independent set and $k' \geq k$. 

First, observe that for every $v_i \in V(G)$ we have that $h(x_iy_i) \in \{ab,cd\}$, so $h(x_i)=a$ if and only if $h(y_i)=b$. This means there are $2k'$ vertices of $G^*$ mapped to $a$ or $b$. The weight of $h$ is at least $2k$, so $k' \geq k$. To show that $I$ is independent, assume that there exists $v_iv_j \in E(G)$ such that $\{v_i, v_j\} \subseteq I$. It means that $h(x_i)=a$ and $h(y_j)=b$, but then the weight of $ab$ in $h$ is $-\infty$, a contradiction. So $G$ has an independent set of size at least $k$.

To complete the proof, note that graph $G^*$ has $2n$ vertices.
\end{proof}

Note that Theorem \ref{thm-homo-c4} implies the hardness for all graphs $H$ that contain $C_4$ as a subgraph. It is because we can set the edge weights related to vertices and edges which do not belong to $C_4$ to $-\infty$. This completes the proof of \autoref{thm-mainhom} (b).


\section{Locally injective and locally bijective homomorphism}\label{sec:lih}
Now let us turn our attention to two other variants of the graph homomorphism problem, i.e., locally injective and locally bijective homomorphism. Recall that for a fixed $H$, the \lihomo{$H$} (\lbhomo{$H$}, resp.) problem asks if a given graph $G$ admits a homomorphism $h$ to $H$ with a restriction that for every $v \in V(G)$, the mapping $h$ is injective (bijective, resp.) on the set $N_G(v)$. Local injectivity can be equivalently seen as ``no two vertices of $G$ with a common neighbor may be mapped to the same vertex of $H$''. Moreover, every locally bijective homomorphism is also locally injective.

We show that unlike the \whomo{$H$}, both \lihomo{$H$} and \lbhomo{$H$} can be solved in subexponential time on string graphs for every $H$.
The crucial observation is all yes-instances have bounded degree.
\mainlihom*
\begin{proof} First, we prove the statement for \lihomo{$H$}, and then for \lbhomo{$H$}.

\paragraph{Locally injective homomorphisms.}
To ensure consistency of solutions found in recursive calls, we will solve a more general problem. First, every vertex $v$ of $G$ is equipped with a list $L(v)$ of vertices of $H$, and we ask for a locally injective homomorphism respecting these lists. Second, we are given a subset $X$ of vertices of $G$ and a function $\sigma \colon X \to 2^{V(H)}$. We require that in a solution $h$, for every vertex $v \in X$ it holds that $h(N_G(v)) = \sigma(v)$, i.e., $\sigma(v)$ is the set of colors appearing on the neighbors of $v$. Clearly, if $L(v) = V(H)$ for every $v$, and $X=\emptyset$, then we obtain the \lihomo{$H$} problem. The algorithm will be recursive and the constrains related to the set $X$ will be checked at leaves of tree of recursive calls.

Observe that if $G$ has a vertex with degree larger than $|H|$, there is no way to map in an injective way. Thus in this case we immediately report a no-instance. So let us assume that every vertex has degree at most $|H|$, which means that $G$ has $O(|H|n)$ edges and thus, by Theorem \ref{thm-stringsep}, there is a balanced separator $S$ of size $O(\sqrt m) = O(\sqrt{n})$, which can be found in time $2^{O(\sqrt{n} \log n)}$ by exhaustive search or in polynomial time, using the geometric representation.
Let $V_1,V_2$ be sets such that $V(G) = V_1 \uplus V_2 \uplus S$, there is no $V_1$-$V_2$-path in $G-S$, and $V_1,V_2 \leq c' \cdot n$ for a constant $c'$.

Let $h$ be a fixed (unknown) solution. For each $v \in S$, we exhaustively guess the color $h(v)$ (respecting the list $L(v)$) and the sets $\sigma_1 := h(N_G(v) \cap V_1)$ and $\sigma_2 := h(N_G(v) \cap V_2)$, respecting $\sigma(v)$, if $v \in X$. Note that, as $h$ is locally injective, $\sigma_1$ and $\sigma_2$ are disjoint. Then recursively solve the problem in $G_1 := G[V_1 \cup S]$ and $G_2 := G[V_2 \cup S]$. In the recursive calls we set $L(v) = \{h(v)\}$ for every $v \in S$. Moreover, we include each such $v$ in $X$ and set $\sigma(v)=\sigma_1$ in the recursive call for $G_1$ and $\sigma(v)=\sigma_2$ in the call for $G_2$. This ensures that no vertex from $S$ has two neighbors of the same color, one in each of $G_1,G_2$. The total number of recursive calls is at most $(|H| \cdot 3^{|H|})^{|S|} = 2^{\tilde{O} (\sqrt{n})}$ and the complexity of the whole algorithm is also $2^{{O} (\sqrt{n} \log n)}$.

\paragraph{Locally bijective homomorphisms.}
The proof follows very similarly to the previous case. It varies only in the part in which we guess the colors of vertices belonging to the separator $S$. Again, let $h$ be a fixed, unknown solution. Observe that as $h$ is locally bijective, each vertex $v$ of $G$ must be mapped by $h$ to a vertex of equal degree. To ensure that, when guessing the color of $v \in S$, we just remove from $L(v)$ all elements $a$ for which $d_G(v)\neq d_H(a)$.
\end{proof}

As mentioned, locally injective homomorphisms generalize some well-studied graph labeling problems, so \autoref{thm-mainlihom} implies the following.

\begin{corollary}
For any fixed $k$,
\begin{enumerate*}[label=(\roman*)] 
\item the $k$-$L(2,1)$-labeling,
\item the $k$-coloring of the square of a graph,
\item the injective $k$-coloring,
\end{enumerate*}
can be solved in time $2^{\tilde{O}(\sqrt{n})}$ in string graphs. \qed
\end{corollary}

On the other hand, as every planar graph is a segment graph \cite{ChalopinG09,DBLP:conf/soda/GoncalvesIP18}, hardness results for planar graphs can be used to derive ETH-lower bounds for \lihomo{$H$} -- in particular, the following ones follow from the hardness results for $k$-$L(2,1)$-labeling~\cite{DBLP:journals/dam/EggemannHN10}, 7-coloring of the square of a graph \cite{227702}, and injective 3-coloring \cite{413218}.
\begin{theorem}[Eggemann {\em et al.} \cite{DBLP:journals/dam/EggemannHN10}, Ramanathan, Lloyd \cite{227702}, Bertossi, Bonuccelli~\cite{413218}]
Unless the ETH fails, there is no algorithm for \lihomo{$H$} in segment graphs working in time $2^{o(\sqrt{n})}$,
where $H$ is 
\begin{enumerate*}[label=(\roman*)]
\item the complement of a path with at least 4 vertices,
\item a complete graph with 7 vertices,
\item a triangle with additional loop on its every vertex. \qed
\end{enumerate*}
\end{theorem}


\section{Locally surjective homomorphism}\label{sec:lsh}
In this section we consider the problem of {locally surjective} homomorphism, denoted by \lshomo{$H$}. For a fixed graph $H$, the \lshomo{$H$} problem asks whether a given graph $G$ admits homomorphism to $H$, which is surjective on $N_G(v)$ for every $v \in V(G)$. In other words, if $h(v)=a$, then every neighbor of $a$ must appear on some neighbor of $v$.
If there exists a locally surjective homomorphism from $G$ to $H$, we denote this fact by by $G \tos H$. For a homomorphism $h: G \to H$ we say that a vertex $v$ is {\em happy} if $h(N_G(v))=N_H(h(v))$. Clearly $h$ is locally surjective if every vertex of $G$ is happy.

We aim to prove \autoref{thm:lsh-together}, i.e., prove a dichotomy for simple target graphs $H$ with $\Delta(H) \leq 2$.
Let us start with a simple observation, that will be used many times.

\begin{observation} \label{obs-lsh}
Let $h:G \tos H$ and $\delta(H) \geq 1$. For every $v \in V(G)$ it holds that:
\begin{compactenum}[a)]
\item  $\deg_G(v)\geq \deg_H(h(v))$. In particular, if $\deg_G(v)=1$, then $\deg_H(h(v))=1$, \label{it-obs-lsh-a}
\item if $\deg_G(v)=\deg_H(h(v))=2$, then $h(v_1) \neq h(v_2)$ for distinct neighbors $v_1, v_2$ of~$v$.  \label{it-obs-lsh-b}\qed
\end{compactenum}
\end{observation}

\subsection{Paths} \label{sec-paths}
In this section, the the target graph is a path $P_k$ with consecutive vertices $1,2,\ldots,k$.

First, let us discuss the case, when $k=3$. Let $G$ be an instance of \lshomo{$P_3$}. By \autoref{obs-lsh} \ref{it-obs-lsh-a}) we can assume that $G$ does not have isolated vertices. We can also assume that an input graph $G$ is bipartite with bipartition classes $X$ and $Y$, as otherwise any homomorphism to $P_3$ cannot exist. Moreover, in any homomorphism, one of bipartition classes, say $Y$, will be entirely mapped to $2$. Note that since no vertex is isolated, vertices of $X$ will always be happy. Thus $G \tos P_3$ if and only if one can color vertices of $X$ with two colors (1 and 3), so that every vertex from $Y$ has at least one neighbor in each color. We observe that this is exactly the {\sc Not All Equal Sat} problem, where $G$ is an incidence graph of the input formula.
From this we conclude that \lshomo{$P_3$} does not have a subexponential algorithm in general graphs, but is solvable in polynomial time in planar graphs (since  {\sc Planar Not All Equal Sat} is in P, see Moret~\cite{Moret:1988:PNP:49097.49099}. Moreover, the list variant of \lshomo{$P_3$} in planar graphs in NP-complete, see Dehghan~\cite{Dehghan2016}.

The win-win approach of \autoref{thm-mainhom} (a) cannot be directly applied for \lshomo{$P_3$}, as there is no good branching on a high-degree vertex. Instead, we will use the following result.
\begin{theorem}[Lee~\cite{Lee16}]\label{thm-stringedges}
There is a constant $c >0$ such that for every $t \geq 1$, it holds that every $K_{t,t}$-free string graph on $n$ vertices has at most $c\cdot n \cdot t\log{t}$ edges.
\end{theorem}
Now we can show that \lshomo{$P_3$} can be solved in subexponential time in string graphs. 

\begin{theorem} \label{thm:lsh-p3}
\lshomo{$P_3$} can be solved in time $2^{O(n^{2/3}\log^{3/2} n)}$ for string graph on $n$ vertices, even if geometric representation is not given.
\end{theorem}

\begin{proof}
We assume an instance graph $G$ has no isolated vertices and is bipartite, with bipartition classes $X$ and $Y$.
Note that $G$ is a yes-instance if and only if there is a homomorphism $h_X \colon G \tos P_3$, such that $h_X(X)=\{1,3\}$ and $h_X(Y)=\{2\}$, or homomorphism $h_Y \colon G \tos P_3$, such that $h_Y(Y)=\{1,3\}$ and $h_Y(X)=\{2\}$.
Let us assume that $X$ is mapped to $\{1,3\}$, the algorithm will be called twice with roles of $X$ and $Y$ switched. Again, we will solve a more general problem, in which we define an additional function $\sigma: Y \to 2^{\{1,3\}}$ and ask for an existence of a homomorphism $h \colon G \to P_3$, such that $\sigma(y) \subseteq h(N_G(y))$ for every $y \in Y$. Clearly, if $\sigma \equiv \{1,3\}$, then we obtain the \lshomo{$P_3$} problem.

First, we observe that if $G$ has at most $\frac{c}{3} \cdot n^{4/3}\log{n}$ edges (where $c$ is a constant from  \autoref{thm-stringedges}), then we can find a balanced separator $S$ of size $O(n^{2/3}\log^{1/2}{n})$ in time $2^{O(n^{2/3}\log^{3/2}{n})}$.  Denote by $V_1,V_2$ the sets such that $V(G) = V_1 \uplus V_2 \uplus S$ and there is no $V_1$-$V_2$-path in $G-S$. We exhaustively guess $h(x)$ for every $x \in S \cap X$ and the partition $\sigma_1 \uplus \sigma_2$ of $\sigma(y) \setminus h(N_G(y) \cap S)$ for every $y \in S \cap Y$. Then, for every possibility, we consider graphs $G_1:=G[V_1 \cup S]$ and $G_2:=G[V_2 \cup S]$, in which vertices of $S$ are already colored. For every $y \in S \cap Y$ we set $\sigma(y)=\sigma_1$ if $y \in V(G_1)$ or $\sigma(y)=\sigma_2$ otherwise.
Then, for every $x \in S \cap X$, we remove $h(x)$ from $\sigma(y)$, for every neighbor $y$ of $x$, and finally we remove $x$ from the instance. If there exists $y$ for which $\sigma(y)=\phi$, we also remove $y$. Then, if any isolated vertex $x \in X$ appears, we remove it too, as it means that $\sigma(y)$ of its every neighbor $y$ was already empty, so the color of $x$ does not matter.
We call the algorithm recursively for graphs $G_1$ and $G_2$, together with their corresponding functions $\sigma$.
Note that $G_1$ or $G_2$ may contain an isolated vertex $y \in Y$ with  $\sigma(y)\neq \phi$, in this case we terminate the current recursive call. Observe that the total number of recursive calls is $2^{|X \cap S|}\cdot 4^{|Y \cap S|}= 2^{{O}(n^{2/3} \log^{1/2}n)}$ and the overall complexity of this step is  $2^{O(n^{2/3}\log^{3/2} n)}$.

If $G$ has more than $\frac{c}{3} \cdot n^{4/3}\log{n}$ edges, we know from Theorem \ref{thm-stringedges} that it has a bipartite subgraph $K_{n^{1/3},n^{1/3}}$. We find it exhaustively in time $n^{O(n^{1/3})} = 2^{O(n^{1/3}\log{n})}$. Let $X' \subseteq X$ and $Y' \subseteq Y$ be its bipartition classes. We branch on three possibilities. Either we set $h(x)=1$ for every $x \in X'$, or $h(x)=3$ for every $x \in X'$ or we choose $x_1, x_2 \in X'$ and set $h(x_1)=1$ and $h(x_2)=3$. In first two cases we can proceed to the graph $G - X'$ (and remove $h(X')$ from $\sigma(y)$ of every $y \in N(X')$), and in the last case we can remove $Y'$, together with $x_1$ and $x_2$ (also adjusting the function $\sigma$ for their neighbors), as all elements of $Y'$ are happy. Denote by $F(n)$ the complexity of this step and observe that
\[F(n)\leq 2^{O(n^{1/3}\log{n})} + 2F(n-n^{1/3}) + n^{2/3}F(n-n^{1/3}) \leq 2^{O(n^{2/3}\log{n})},\]
so the total complexity of algorithm is also $2^{O(n^{2/3}\log^{3/2}{n})}$.
\end{proof}

For paths with at least 4 vertices, the existence of  subexponential algorithms are unlikely.

\begin{theorem} \label{thm:lsh-pp-lower}
For any $k \geq 4$, the \lshomo{$P_k$} problem on \kdir{$2$} graphs with $n$ vertices cannot be solved in time $2^{o(n)}$, unless the ETH fails.
\end{theorem}
\begin{proof}
Let $k \geq 4$ be fixed. We reduce from 3-\sat, consider an instance $\Phi$ of 3-\sat with variables $u_1,\dots,u_n$ and clauses $c_1,\dots,c_m$. We assume that every variable $u$ appears at least once as positive and once as negative literal. Indeed, otherwise we can set its value and remove it from the formula, along with all clauses containing $u$.

For each occurrence of a variable $u_i$ in a clause, we introduce a vertical \emph{occurrence segment} $x$. We denote the sets of positive and negative occurrence segments of $u_i$ by $X_i$ and $\tilde{X}_i$ respectively. Let $X=\bigcup_{i \in [k]} (X_i \cup \tilde{X}_i)$. We place the elements of $X$ in a following order: leftmost segments from $X_1$, then the ones from $\tilde{X}_1$, $X_2$, $\tilde{X}_2$, etc. Moreover, each segment is slightly shorter than the one on its left (see Figure \ref{fig-lsh-construction} (a)). For each clause $c_p$, we introduce a horizontal \emph{clause segment} $y_p$, intersecting all segments from $X$. Let $Y=\bigcup_{p \in [m]}y_p$. We also we add a vertical path $T=(t_1, \dots,t_{2k-1})$ on the right side of the picture. For each $y \in Y$ we introduce a horizontal segment $y'$ on the right side of $y$, which intersects $y$ and $t_1$. Let $Y'$ be the  set of all these segments $y'$. For every $x \in X$ we add a horizontal path $(q_1, \dots, q_{k-3})$ such that $q_1$ intersects $x$ and $q_{k-3}$ intersects $t_k$.

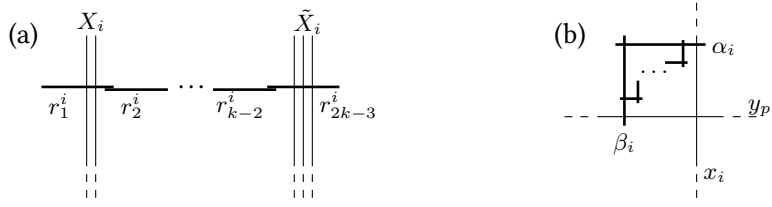
\begin{figure}[h]
\centering\begin{tikzpicture}[scale=1.2]
\draw[dashed] (-1.3,0.2) -- (-1.3,0.5);
\draw (-1.3,2) -- (-1.3,0.55);
\draw[dashed] (-1.2,0.2) -- (-1.2,0.5);
\draw (-1.2,2) -- (-1.2,0.55);
\foreach \i in {0,...,2}
{\draw[dashed] (1+0.1*\i,0.2) -- (1+0.1*\i,0.5);
\draw (1+0.1*\i,2) -- (1+0.1*\i,0.55);}
\node[draw=none,fill=none] at (-2,2) {(a)};
\draw[line width=1] (-1.8,1.43) -- (-1,1.43);
\draw[line width=1] (-1.1,1.4) -- (-.4,1.4);
\draw[line width=1] (.7,1.43) -- (1.5,1.43);
\draw[line width=1] (.1,1.4) -- (.8,1.4);
\node at (-.12,1.43) {\dots};
\node[draw=none,fill=none] at (-.8,1.22) {\footnotesize{$r_2^i$}};
\node[draw=none,fill=none] at (0.4,1.22) {\footnotesize{$r_{k-2}^i$}};
\node[draw=none,fill=none] at (-1.6,1.22) {\footnotesize{$r_1^i$}};
\node[draw=none,fill=none] at (1.6,1.22) {\footnotesize{$r_{2k-3}^i$}};
\node at (-1.25,2.15) {\footnotesize{$X_i$}};
\node at (1.15,2.15) {\footnotesize{$\tilde{X}_i$}};
\end{tikzpicture}
\hskip 2cm
\centering\begin{tikzpicture}[scale=1.2]
\draw[dashed] (-0.6,2.1) -- (-0.6,2.4); 
\node[draw=none,fill=none] at (-0.4,0.45) {\footnotesize{$x_i$}};
\draw[dashed] (-0.6,0.2) -- (-0.6,0.5);
\draw (-0.6,2) -- (-0.6,0.55);
\node[draw=none,fill=none] at (-2,2) {(b)};
\draw (-0.3,1.1) -- (-1.7,1.1);
\draw[dashed] (-0.2,1.1) -- (0.1,1.1);
\draw[dashed] (-1.8,1.1) -- (-2.1,1.1);
\node[draw=none,fill=none] at (0.1,1.2) {\footnotesize{$y_p$}};
\draw[line width=1] (-0.5,1.9) -- (-1.5,1.9);
\draw[line width=1] (-1.4,1) -- (-1.4,2);
\node[draw=none,fill=none] at (-1.4,0.8) {\footnotesize{$\beta_i$}};
\node[draw=none,fill=none] at (-0.3,1.85) {\footnotesize{$\alpha_i$}};
\draw[line width=1] (-1.45,1.3) -- (-1.2,1.3);
\draw[line width=1] (-1.25,1.25) -- (-1.25,1.5);
\draw[line width=1] (-0.95,1.7) -- (-0.7,1.7);
\draw[line width=1] (-0.75,1.65) -- (-0.75,1.95);
\node[draw=none,fill=none] at (-1.05,1.6) {\footnotesize{$\dots$}};
\end{tikzpicture}
\caption{Construction of gadgets for both Theorems \ref{thm:lsh-pp-lower} and \ref{thm:lsh-cp-lower}. (a) A variable gadget for $H=P_k$. If $H=C_k$, then the path should contain $k-1$ segments. (b) A membership gadget. The path between $\alpha_i$ and $\beta_i$ contains $2k-4$ segments if $H=P_k$ or $k-2$ segments if $H=C_k$. If $H=C_3$, the only segment from the path is placed on $\alpha_i$ in a way that it intersects $\beta_i$ but not $x_i$.}
\label{fig-cthulhu}
\end{figure}

For each $u_i$ we add a \emph{variable gadget} on top of its occurrence segments. It is a path $(r_1^i, \dots,r_{2k-3}^i)$ such that $r_1^i$ ($r_{2k-3}^i$, resp.) intersects all segments from $X_i$ ($\tilde{X}_i$, resp., see Figure \ref{fig-cthulhu} (a)).
Now consider an occurrence segment $x_i$, and let $c_p$ be the clause containing this particular occurrence of a variable. On the intersection of $x_i$ with $y_p$ we introduce a \emph{membership gadget}, containing two segments $\alpha_i$ and $\beta_i$, crossing each other and $x_i$ or $y_p$, resp. Also, we add a path $(s_1, \dots,s_{2k-4})$ such that $s_1$ and $s_{2k-4}$ also intersect $\alpha_i$ and $\beta_i$, resp. (see Figure \ref{fig-cthulhu} (b)). Let $E^i$ be the set of segments of the membership gadget of~$x_i$. 

Assume that there exists $h: G \tos P_k$. We will show the satisfying truth assignment $\varphi$ of $\Phi$. First, note that without loss of generality $h(t_{2k-1})=1$, as $\deg(t_{2k-1})=1$. From Obs. \ref{obs-lsh} b) we get that $h(t_k)=k$ and $h(t_1)=1$. It implies that $h(Y')=2$ and, by Obs. \ref{obs-lsh} b), $h(Y)=3$. Moreover, as $h(t_k)=k$, then for every $x \in X$ its corresponding path between $x$ and $t_k$ can be colored in only one way, and $h(X)=2$.
Also, for every $v_i \in V(G)$ it holds that $\{h(r_1^i), h(r_{2k-3}^i)\}=\{1,3\}$. Define $\varphi$ as follows: if $h(r_1^i)=1$, then $\varphi(u_i)=1$, if $h(r_{2k-3}^i)=1$, then $\varphi(u_i)=0$. Assume that there exists a clause $c_p$, which is not satisfied by $\varphi$, and let $x_i, x_j, x_l$ be the occurrence segments corresponding to the literals of $c_p$. As all literals of $c_p$ are false, the neighbors of $x_i, x_j, x_l$ in their variable gadgets are mapped to 3. To make $x_i,x_j,x_l$ happy, we need to have $h(\alpha_i)=h(\alpha_j)=h(\alpha_l)=1$, which implies that  $h(\beta_i)=h(\beta_j)=h(\beta_l)=2$. But then $y_p$ is not happy, a contradiction.

Now we show that if $\Phi$ has a satisfying assignment $\varphi$, then there exists $h: G \tos P_k$. We set the coloring $h$ of all segments of $G$, which do not belong to variable or membership gadgets, exactly like in previous step. If $\varphi(u_i)=1$, we set $h(r_1^i)=1$ and $h(r_{2k-3}^i)=3$. Otherwise, $h(r_1^i)=3$ and $h(r_{2k-3}^i)=1$. In both there is only way to color the path $(r_2^i,\dots,r_{2k-4}^i)$ in a locally surjective way.

Observe that every $x_i \in X$ which is adjacent to $r_1^i$ if $\varphi(u_i)=1$ or to $r_{2k-3}^i$ if $\varphi(u_i)=0$ is already happy.
In this case we set $h(\alpha_i)=3$ and $h(\beta_i)=4$, and color $s_1,s_2,\ldots,s_{2k-4}$ with $2,1,2\ldots,k-1,k,k-1,\ldots,5$ (or $2,1,2,3$ for $k=4$). If the neighbor of $x_i \in X$ is not mapped to 1 (i.e., $x_i$ corresponds to a false literal), we must make it happy by setting $h(\alpha_i)=1$, which implies that $h(\beta_i)=2$. We color the path $s_1,\ldots,s_{2k-4}$ with $2,3,\ldots,k-1,k,k-1,\ldots,3$. This makes all segments inside $E^i$ happy.

Finally, observe that each since $\varphi$ is a satisfying assignment, each clause $c_p$ contains a true literal $x_i$, and thus $y_p$ has a nieghbor $\beta_i$ mapped to 4, so each variable segment is happy. This means that $h$ is a locally surjective homomorphism.
\end{proof}
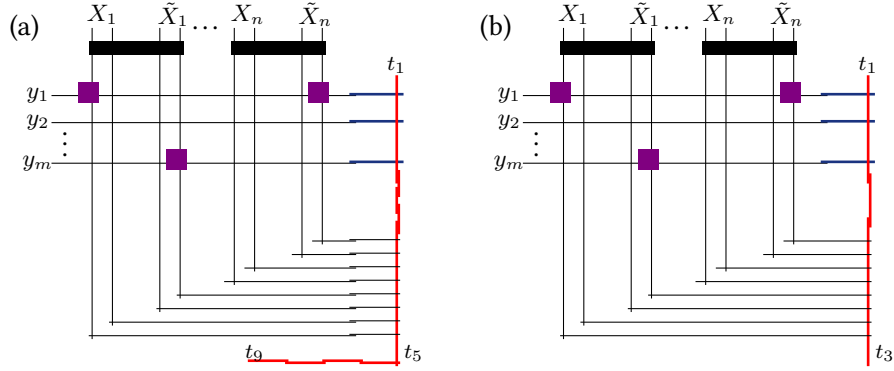
\begin{figure}[h]
\centering\begin{tikzpicture}[scale=.9]
\node[draw=none,fill=none] at (-2.4,3.2) {(a)};
\draw (-1.1,-1.2) -- (-1.1,3.2);
\draw (-0.1,-0.8) -- (-0.1,3.2);
\draw (-0.4,-1) -- (-0.4,3.2);
\draw (-1.4,-1.4) -- (-1.4,3.2);
\draw (0.7,-0.6) -- (0.7,3.2);
\draw (1.7,-0.2) -- (1.7,3.2);
\draw (1,-0.4) -- (1,3.2); 
\draw (2,0) -- (2,3.2);
\draw (-2,1.2) -- (2.5,1.2); \draw[line width=1,color=blue] (3.2,1.22) -- (2.4,1.22);
\draw (-2,1.8) -- (2.5,1.8); \draw[line width=1,color=blue] (3.2,1.82) -- (2.4,1.82);
\draw (-2,2.2) -- (2.5,2.2); \draw[line width=1,color=blue] (3.2,2.22) -- (2.4,2.22);
\node at (-1.8,1.6) {\vdots};
\node at (0.3,3.2) {\dots};
\node at (-1.25,3.4) {\footnotesize{$X_1$}};
\node at (-0.2,3.4) {\footnotesize{$\tilde{X}_1$}};
\node at (0.9,3.4) {\footnotesize{$X_n$}};
\node at (1.9,3.4) {\footnotesize{$\tilde{X}_n$}};
\node at (-2.2,1.2) {\footnotesize{$y_m$}};
\node at (-2.2,1.8) {\footnotesize{$y_2$}};
\node at (-2.2,2.2) {\footnotesize{$y_1$}};
\draw[line width=1,color=red] (3.1,.9) -- (3.1,2.5);
\draw[line width=1,color=red] (3.13,1.1) -- (3.13,0.7); 
\draw[line width=1,color=red] (3.1,0.35) -- (3.1,-1.8); 
\draw[line width=1,color=red] (3.1,0.45) -- (3.1,0.85); 
\draw[line width=1,color=red] (3.13,0.6) -- (3.13,0.15);
\node at (3.1,2.65) {\scriptsize{$t_1$}};
\node at (3.35,-1.6) {\scriptsize{$t_5$}};
\fill[black] (-1.45,2.8) rectangle (-.05,3);
\fill[black] (0.65,2.8) rectangle (2.05,3);
\fill[violet] (-0.3,1.1) rectangle (0,1.4);
\fill[violet] (-1.6,2.1) rectangle (-1.3,2.4);
\fill[violet] (1.8,2.1) rectangle (2.1,2.4);
\foreach \i in {0,...,3}
{\draw (-1.45+\i,-1.35+0.4*\i) -- (2.5,-1.35+0.4*\i); \draw (2.4,-1.33+0.4*\i) -- (3.15,-1.33+0.4*\i);
\draw (-1.15+\i,-1.15+0.4*\i) -- (2.5,-1.15+0.4*\i); \draw (2.4,-1.13+0.4*\i) -- (3.15,-1.13+0.4*\i);}
\draw[line width=1,color=red] (2.55,-1.75) -- (3.15,-1.75); 
\draw[line width=1,color=red] (2,-1.72) -- (2.6,-1.72); 
\draw[line width=1,color=red] (1.45,-1.75) -- (2.05,-1.75); 
\draw[line width=1,color=red] (1.5,-1.72) -- (0.9,-1.72); 
\node at (1,-1.6) {\scriptsize{$t_9$}};
\end{tikzpicture}
\hskip .4cm
\centering\begin{tikzpicture}[scale=.9]
\node[draw=none,fill=none] at (-2.4,3.2) {(b)};
\draw (-1.1,-1.2) -- (-1.1,3.2);
\draw (-0.1,-0.8) -- (-0.1,3.2);
\draw (-0.4,-1) -- (-0.4,3.2);
\draw (-1.4,-1.4) -- (-1.4,3.2);
\draw (0.7,-0.6) -- (0.7,3.2);
\draw (1.7,-0.2) -- (1.7,3.2);
\draw (1,-0.4) -- (1,3.2); 
\draw (2,0) -- (2,3.2);
\draw (-2,1.2) -- (2.5,1.2); \draw[line width=1,color=blue] (3.2,1.22) -- (2.4,1.22);
\draw (-2,1.8) -- (2.5,1.8); \draw[line width=1,color=blue] (3.2,1.82) -- (2.4,1.82);
\draw (-2,2.2) -- (2.5,2.2); \draw[line width=1,color=blue] (3.2,2.22) -- (2.4,2.22);
\node at (-1.8,1.6) {\vdots};
\node at (0.3,3.2) {\dots};
\node at (-1.25,3.4) {\footnotesize{$X_1$}};
\node at (-0.2,3.4) {\footnotesize{$\tilde{X}_1$}};
\node at (0.9,3.4) {\footnotesize{$X_n$}};
\node at (1.9,3.4) {\footnotesize{$\tilde{X}_n$}};
\node at (-2.2,1.2) {\footnotesize{$y_m$}};
\node at (-2.2,1.8) {\footnotesize{$y_2$}};
\node at (-2.2,2.2) {\footnotesize{$y_1$}};
\draw[line width=1,color=red] (3.1,.9) -- (3.1,2.5); 
\draw[line width=1,color=red] (3.1,0.4) -- (3.1,-1.8); 
\draw[line width=1,color=red] (3.13,0.25) -- (3.13,1.05); 
\node at (3.1,2.65) {\scriptsize{$t_1$}};
\node at (3.35,-1.6) {\scriptsize{$t_3$}};
\fill[black] (-1.45,2.8) rectangle (-.05,3);
\fill[black] (0.65,2.8) rectangle (2.05,3);
\fill[violet] (-1.6,2.1) rectangle (-1.3,2.4);
\fill[violet] (1.8,2.1) rectangle (2.1,2.4);
\fill[violet] (-0.3,1.1) rectangle (0,1.4);
\foreach \i in {0,...,3}
{\draw (-1.45+\i,-1.35+0.4*\i) -- (3.15,-1.35+0.4*\i);
\draw (-1.15+\i,-1.15+0.4*\i) -- (3.15,-1.15+0.4*\i);}
\end{tikzpicture}
\caption{An overall construction of $G$ for (a) path $P_5$ and (b) cycle $C_5$. The red segments stands for $T$, blue ones are the elements of $Y'$. Black rectangles are variable gadgets with details shown on Figure \ref{fig-cthulhu} (a). Violet squares are membership gadgets with details shown on Figure \ref{fig-cthulhu} (b).}
\label{fig-overall}
\end{figure}

\subsection{Cycles}
\label{sec-cycles}
In this section we assume that the target graph is a cycle with consecutive vertices $1,2,\ldots,k$ (so $1$ is adjacent to $k$). Let us start with the case $k=4$.

\begin{theorem} \label{thm:lsh-c4}
\lshomo{$C_4$} can be solved in time $2^{O(n^{2/3}\log^{3/2} n)}$ for string graph on $n$ vertices, even if a geometric representation is not given.
\end{theorem}
\begin{proof}
Again, we assume that an instance graph $G$ is bipartite with bipartition classes $X$ and $Y$, without isolated vertices. Clearly in any solution $h$ we either have $h(X) = \{1,3\}$ and $h(Y)=\{2,4\}$, or $h(X) = \{2,4\}$ and $h(Y)=\{1,3\}$.

Let us show that there exists $h:G \tos C_4$ such that $h(X)=\{1,3\}$ if and only if there exist $h_1,h_2:G \tos P_3$ for which $h_1(X)=h_2(Y)=\{1,3\}$ and $h_1(Y)=h_2(X)=\{2\}$.
First, consider $h: G \tos C_4$ such that $h(X)=\{1,3\}$, which implies that $h(Y)=\{2,4\}$. Define $h_1(z) = 2$ if $z \in Y$ and $h_1(z)=h(z)$ otherwise. Define $h_2(z)=2$ if $z \in X$ and $h_1(z)=h(z)-1$ otherwise. Clearly, $h_1,h_2 \colon G \to P_3$ and $h_1(X)=h_2(Y)=\{1,3\}$. Assume $h_1$ is not locally surjective, i.e., there exists $v$ which is not happy. Note that if $v \in X$, then $v$ must be an isolated vertex, a contradiction. If $v \in Y$ and, without loss of generality, $h_1(N(v))=\{1\}$, this means that $h(N(v))=\{1\}$, again, a contradiction, because $h$ is locally surjective. Analogous argument works for $h_2$.

Now assume that there exist $h_1,h_2 :G \tos P_3$ such that $h_1(X)=h_2(Y)=\{1,3\}$. For every $z \in X$ let $h(z):=h_1(z)$ and for every $z \in Y$ let $h(z):=h_2(z)+1$. Observe that $h(X) = \{1,3\}$ and $h(Y)=\{2,4\}$, so $h$ is a homomorphism. Assume it is not locally surjective, and $v$ is not happy. Without loss of generality let $h(v)=2$ and $h(N(v))=\{1\}$. This means that $h_1(v)=2$ and, as $h_1$ was locally surjective, there is $u \in N(v)$ for which $h_1(u)=3$. But $u \in X$, so $h_1(u)=h(u)=3$, a contradiction.

To solve \lshomo{$C_4$}, we run the algorithm from Theorem \ref{thm:lsh-p3} twice, switching the roles of $X$ and $Y$. We return true for \lshomo{$C_4$} if both calls return true.

The total running time is $2^{O(n^{2/3}\log^{3/2} n)}$.
\end{proof}

It appears that existence of subexponential algorithms for remaining~$k$ is unlikely.

\begin{theorem} \label{thm:lsh-cp-lower}
Let $k \geq 3$, $k \neq 4$. There is no algorithm solving \lshomo{$C_k$} on a \kdir{$2$} graph with $n$ vertices in time $2^{o(n)}$, unless the ETH fails.
\end{theorem}
\begin{proof}
Again, we reduce from 3-\sat. Assume that every variable $u$ appears at least once as a positive and once as a negative literal. Our construction is similar to the one in the proof of Theorem \ref{thm:lsh-pp-lower}, we just change the length of some paths. We construct the sets $X$, $Y$ and $Y'$ in the same way. Each variable gadget is now a path $(r_1^i,\dots,r_{k-1}^i)$ such that $r_1^i$ and $r_{k-1}^i$ intersects also all segments from $X_i$ and $\tilde{X}_i$ respectively. In each membership gadget $E^i$ the path between $\alpha_i$ and $\beta_i$ has now $k-2$ segments. Also the path $T$ has length $k-2$ and still each $y' \in Y'$ intersects $t_1$. For every $x \in X$, we add a single segment $x'$, intersecting $x$ and $t_{k-2}$ (instead of a path, see Figure~\ref{fig-cthulhu}~(a)). Denote the set of these segments $x'$ by $X'$.

We show that $G \tos C_k$ iff $\Phi$ is satisfiable. Assume that there exists $h \colon G \tos C_k$. Let us start with analyzing its structure, the argument will be split into two cases.

First consider $k\geq 5$. By symmetry of $C_k$, we can assume that $h(t_1)=1$ and $h(t_2)=k$. By \autoref{obs-lsh} b), it implies that $h(t_{k-2})=4$. Vertices $t_1$ and $t_{k-2}$ must be happy, so there exist $x'\in X'$ and $y' \in Y'$, such that $h(x')=3$ and $h(y')=2$, which means $h(x)=2$ and $h(y)=3$ for their corresponding neighbors $x \in X$ and $y \in Y$. Note that if there exists $z' \in X'$ such that $h(z')=5$, then its neighbor $z$ from $X$ must be mapped to 6 (or 1, if $k=5$). However, then $h(x)$ is not a neighbor of $h(y)=3$, a contradiction.

If $k=3$, then, by symmetry, we assume that $h(t_1)=1$ and $h(y')=2$ for some $y' \in Y'$. By Obs. \ref{obs-lsh} b), note that $h(y)=3$ for the neighbor $y$ of $y'$. If there is some $x' \in X'$, such that $h(x')=2$, then its neighbor $x$ must be mapped to 3, which is impossible since $h(y)=3$.

In both cases we obtain that every segment from $X'$ is mapped to $3$ and every segment from $X$ is mapped to 2. Analogously we can show that $h(Y')=2$ and $h(Y)=3$.

For each $i$, we have $\{h(r_1^i),h(r_{k-1}^i)\}=\{1,3\}$.
We define $\varphi(u_i)=1$ if $h(r_1^i)=1$, otherwise $\varphi(u_i)=0$.
Suppose that $\varphi$ does not satisfy $\Phi$. Let $c_p$ be an unsatisfied clause. Since $h$ is locally surjective, $y_p$ is happy, so it has a neighbor $\beta_i$ such that $h(\beta_i)=4$ (or $h(\beta_i)=1$ if $k=3$). It implies that  $h(\alpha_i)=3$ and thus the neighbor of the occurrence segment $x_i$ in the variable gadget is mapped to 1. Therefore $x_i$ corresponds to a true literal, a contradiction.

Now assume that $\varphi$ is a satisfying assignment for $\Phi$. We define the coloring $h$ of all vertices of $G$ except the ones in variable or membership gadgets exactly as above. For each variable $u_i$, if $\varphi(u_i)=1$, we set $h(r_1^i)=1$, otherwise $h(r_{k-1}^i)=1$. We color remaining vertices of vertex gadgets in the only possible way. Observe that every $x_i \in X$ which has a neighbor $r_1^i$ or $r_{k-1}^i$ mapped to $1$ is already happy, so we can set $h(\alpha_i)=3$ and $h(\beta_i)=4$ (or $h(\beta_i)=1$ if $k=3$), and color all remaining vertices of $E^i$ such that $h(E^i)=[k]$. Such vertices $x_i$ corresponds to true literals.
If $x_i$ still has no neighbors mapped to 1, we need to set $h(\alpha_i)=1$, which implies $h(\beta_i)=2$. Note that this coloring can be extended to the remaining segments in the membership gadget.
Observe that a clause segment $y_p$ is happy only if it has a neighbor $\beta_i$ mapped to $4$ (or 1 for $k=3$), and recall that for such $\beta_i$, the segment $x_i$ corresponds to a true literal. As $\varphi$ is a satisfying assignment, such literal exists in each clause, so all vertices of $Y$ must be happy, which means $h$ is locally surjective.
\end{proof}


\subsection{One more hard graph}
Finally, let us consider  the graph $H$ in Fig.~\ref{fig:lsh} (left). We will show the following.

\begin{figure}[h]
\centering\begin{tikzpicture}[scale=1, every node/.style={draw,circle,fill=white,inner sep=0pt,minimum size=5pt}]
\draw[line width=1] (0,0) -- (2,0);
\node at (0,0) {} ;
\node at (2,0) {} edge [line width=1,in=50,out=130,loop] ();
\node[draw=none,fill=none] at (0,-0.4) {\footnotesize{a}};
\node[draw=none,fill=none] at (2,-0.4) {\footnotesize{b}};
\end{tikzpicture}
\hskip 2cm
\includegraphics[scale=0.8,page=1]{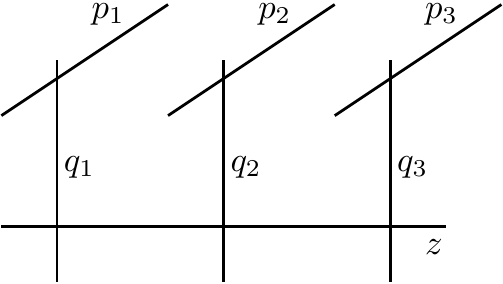}
\caption{A graph $H$ (left) and a clause gadget (right).}
\label{fig:lsh}
\end{figure}
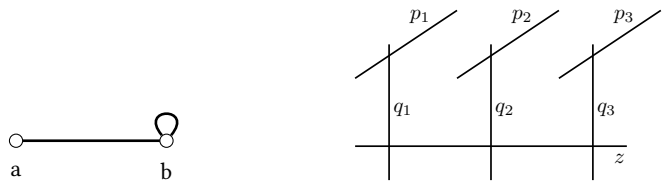

\begin{theorem} \label{thm:lsh-lower}
Assuming the ETH, there is no algorithm solving the \lshomo{$H$} on a segment graph with $n$ vertices in time $2^{o(n)}$.
\end{theorem}
\begin{proof}
We reduce from 3-\sat. Consider a 3-\sat formula $\Phi$ with variables $u_1,u_2,\ldots,u_n$ and clauses $c_1,c_2,\ldots,c_m$, each of which is an alternative of exactly three literals. Again, we may assume that each variable $u$ appears at least once as a positive and at least once as a negative literal.

Let us construct a segment graph $G$, which is an instance of \lshomo{$H$}. For each variable $u_i$ we introduce two {\em variable segments} $x_i$ and $y_i$, intersecting each other. The segment $x_i$ will correspond to positive appearances of $u_i$, while $y_i$ will correspond to the negative ones. For each clause we introduce a {\em clause gadget} depicted in Fig.~\ref{fig:lsh} (right). The segments $p_1,p_2,p_3$ correspond to literals of the clause.

For every appearance of $u_i$ in a clause $c_j$, we add an {\em occurrence segment} intersecting the appropriate segment of $x_i,y_i$ and one of $p_1,p_2,p_3$ in the gadget encoding the clause $c_j$. The occurrence segments do not intersect other variable segments and segments in clause gadgets, but may intersect each other. Finally, for every occurrence segment $s$ we add two segments $e$ and $f$, such that $f$ intersects only $e$, and $e$ intersects only $s$ and $f$. The overall picture of the construction is shown in Fig.~\ref{fig-lsh-construction}.

\begin{figure}[h]
\centering
\includegraphics[scale=0.8,page=2]{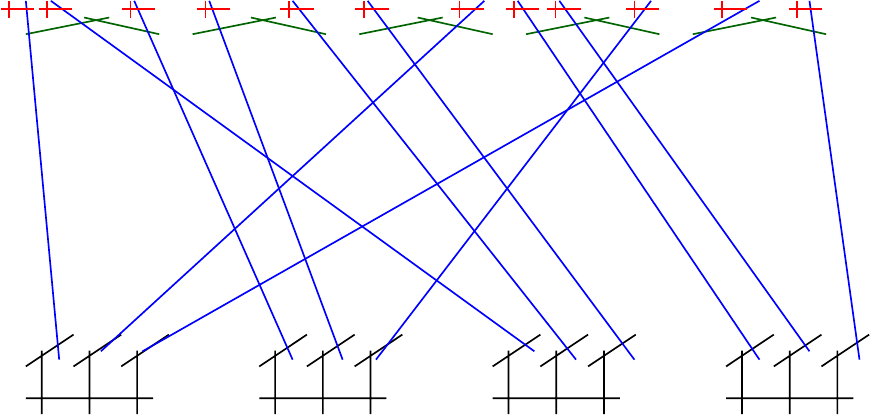}
\caption{A construction in Theorem \ref{thm:lsh-lower}. Clause gadgets are drawn in black, variable segments in green, occurrence segments in blue, and their corresponding segments in red.}
\label{fig-lsh-construction}
\end{figure}

Let us show that $\Phi$ is satisfiable if and only if $G$ has a locally surjective homomorphism to $H$. First, suppose that $\Phi$ is satisfiable and let $\varphi$ be a satisfying truth assignment. Let us define a mapping $h \colon V(G) \to \{a,b\}$. For each variable $u_i$, if $\varphi(u_i)$ is true, then $h(x_i)=a$ and $h(y_i)=b$, otherwise $h(x_i)=b$ and $h(y_i)=a$. For each occurrence segment $s$, and its {\em corresponding segments} $e,f$, we set $h(s)=b$, $h(e)=b$, and $h(f)=a$.
Now consider a clause gadget corresponding to a clause $c_j$. Since $\varphi$ is a satisfying assignment, $c_j$ has at least one true literal, let is be $k$-th literal in $c_j$ for $k \in \{1,2,3\}$. We set $h(p_k)=b$ and $h(q_k)=a$. Moreover, we set $h(p_i)=a$ and $h(q_i)=b$ for $i \neq k$. Finally, we set $h(z)=b$.
It is straightforward to see that $h$ is a homomorphism to $H$. Let us now show that it is locally surjective.

First, we observe that each variable segment is happy. Indeed, recall that $h(x_i) \neq h(y_i)$ and that each of $x_i,y_i$ intersects at least one occurrence segment, which is mapped to $b$. Now consider an occurrence segment $s$ with its corresponding segments $e,f$. The segment $f$ is happy, because it is adjacent to $e$, which is mapped to $b$. The segment $e$ is also happy, since it is adjacent to $f$ and $s$, which are mapped to $a$ and $b$, respectively. Moreover, $s$ is adjacent to $e$, so to make it happy, it needs to be adjacent to a vertex mapped to $a$.
If the literal corresponding to $s$ is true, then such a vertex is either $x_i$ or $y_i$. So assume that $s$ corresponds to a literal that is false. Note that this literal is not satisfying any clause, so $s$ intersects some $p_k$ in a clause gadget, such that $h(p_k)=a$. Therefore $s$ is always happy. Finally, since each occurrence segment is mapped to $b$, it is straightforward to see that each segment in a clause gadget is also happy. This shows that $h$ is locally surjective.

For the other direction, suppose that $h$ is a locally surjective homomorphism from $G$ to $H$.
Consider an occurrence segment $s$ with its corresponding segments $e,f$. Note that since $f$ is happy, it must that $h(f) = a$ and $h(e)=b$. Now, since $e$ is happy, we must have $h(s)= b$. Now consider the variable segments for a variable $u_i$. Note that they only intersect each other and occurrence segments. Thus, to make them happy, one of $x_i,y_i$ must be mapped to $a$ and the other one to $b$. For each variable $u_i$, we set $\varphi(u_i)$ true if an only if $h(x_i)=a$. Let us show that $\varphi$ satisfies $\Phi$.
Suppose the contrary, i.e., there is a clause $c_j$ which is not satisfied by $\varphi$, i.e., all segments corresponding to literals of $c_j$ are mapped to $b$. Consider the segments in the clause gadget corresponding to $c_j$. Note that in order to make the occurrence segments happy, we need to set $h(p_1)=h(p_2)=h(p_3)=a$. Since $h$ is a homomorphism, we need to have $h(q_1)=h(q_2)=h(q_3)=b$. Now, to make $q$'s happy, we need to have $h(z)=b$. However, this way $z$ is not adjacent to any segment mapped to $a$, so it is unhappy, a contradiction.
\end{proof}


\section{Consequences for $P_t$-free graphs}
Let us start with proving \autoref{thm-mainptfree}.

\ptfreethm*

\begin{proof}
Recall that part (a) was proven by Groenland {\em et al.} \cite{DBLP:journals/corr/abs-1803-05396}. We will show that the proof of \autoref{thm-mainhom} (b) implies \autoref{thm-mainptfree} (b).
Let us consider again the problem \whomo{$H$}, for $H$ shown in Figure \ref{fig-themagnificentseven} (a). We go back to the proof of \autoref{thm-mchomo} and observe that the longest induced path of each instance $G^*$ has at most 6 vertices (if $v_iv_{i'}, v_jv_{j'}$ are disjoint edges of $G$, then it is the path $\beta_{ii'},\alpha_{ii'},x_i,y_j,\alpha_{jj'}, \beta_{jj'})$). Clearly, this means that there is no algorithm solving {\sc Max Cut} (and thus \whomo{$H$} for $H$ in Figure \ref{fig-themagnificentseven} (a)) in time $2^{o(n)}$ on $P_7$-free graphs on $n$ vertices, unless the ETH fails. Moreover, if instead of gadgets we used edge-weights, as in the proof of  \autoref{thm-homo-c4}, we obtain hardness of \whomo{$H$} for complete graphs. Note that complete graphs are $P_3$-free, and clearly the problem is polynomially solvable on $P_2$-free graphs.

Analogously we can conclude that, assuming the ETH, there is no subexponential algorithm for \whomo{$H$} for $H$ shown in Figure \ref{fig-themagnificentseven} (b). Note that the instance constructed in the proof of \autoref{thm-octhomo} is always $P_{13}$-free, and substituting gadgets with appropriate edge weight gives the hardness for complete bipartite graphs.
For the remaining graphs in   Figure \ref{fig-themagnificentseven}, the instance constructed in the proof of \autoref{thm-homo-c4} is also complete bipartite. Note that complete bipartite graphs are $P_4$-free, and for all graphs $H$ in Figures \ref{fig-themagnificentseven} (b) -- (g), \whomo{$H$} is polynomially solvable for $P_3$-free graphs.
\end{proof}

In particular, we obtain the following result, answering an open problem of Bonamy {\em et al.}~\cite{DBLP:journals/algorithmica/BonamyDFJP19}.

\begin{corollary} \label{cor:oct-ptfree}
\oct problem is NP-complete and cannot be solved in time $2^{o(n)}$ in $P_{13}$-free graphs, unless the ETH fails. \qed
\end{corollary}

Bonamy {\em et al.}~\cite{DBLP:journals/algorithmica/BonamyDFJP19} considered also a closely related problem \ioct, where we additionally require that the removed set of vertices is independent. Interestingly, the hardness result of Corollary \ref{cor:oct-ptfree} does not carry over to this problem. Indeed, \ioct is equivalent to finding a 3-coloring of the input graph, in which the size of one color class is minimized. It is straightforward to see that this problem can be stated as \whomo{$K_3$}, where the weight associated with one vertex is 0, the weights associated with two other vertices are 1, and all edge weights are 0. Thus, by \autoref{thm-mainptfree}, we obtain the following.

\begin{corollary} \label{cor:ioct-ptfree}
For every fixed $t$, the \ioct problem can be solved in time $2^{O(\sqrt{n})}$ for $P_t$-free graphs on $n$ vertices. \qed
\end{corollary}

Let us also point out that applying the approach of Theorem \ref{thm-mainlihom} to a $P_t$-free graph yields a polynomial algorithm (for fixed $H$). Indeed, a $P_t$-free graph with maximum degree at most $|H|$ has at most $t \cdot |H|^t$ vertices, which is a constant, and thus the problem can be brute-forced in constant time.

Moreover, we observe that also proofs in the Section \ref{sec:lsh} give corollaries for $P_t$-free graphs. Indeed, the graphs constructed in \autoref{thm:lsh-together} (b) are $P_t$-free for some $t$ (depending on $k$). The longest induced path in the graph constructed in the Case 1 of the proof has at most $10k$ vertices: it contains four vertices from $X$, $2(k-1)$ vertices from two variable gadgets, $2(k-3)$ vertices from two paths joining elements of $X$ with $t_k$, $2(2k-3)$ vertices from two membership gadgets and the vertex $t_k$ itself. The longest induced path in  graph constructed in the Case 2 of the proof has $4k+3$ vertices: again, four vertices from $X$, $2(k-1)$ vertices from two variable gadgets, two vertices from $X'$, $2(k-1)$ vertices from two membership gadgets and the vertex $t_k$ itself. From this we conclude that if $H$ is an irreflexive graph with $\Delta(G)\leq 2$ then for some constant $t$ the subexponential algorithm for \lshomo{$H$} for $P_t$-free graphs does not exist, unless the ETH fails.
Finally, note that the construction in the proof of \autoref{thm:lsh-lower} can be modified so that all vertices corresponding to occurrence segments form a clique. After this modification the graph might not be a segment graph anymore, but it is $P_{12}$-free.


\section{Further research directions}
Let us conclude the paper with pointing out some directions for further research.
First, it would be interesting to obtain a complexity dichotomy for the problems of finding a homomorphism and a list homomorphism from a string graph to a fixed graph $H$.
Next, we think that obtaining a full complexity dichotomy for \lshomo{$H$}  in string graphs is an exciting (and probably difficult) task. Let us mention that the NP-hardness proof by Fiala and Paulusma \cite{DBLP:journals/tcs/FialaP05} implies that if $H$ is a connected graph with at least two edges, then \lshomo{$H$} cannot be solved in subexponential time in general graphs.

Finally, recall that our hardness proofs for \lshomo{$H$} imply hardness for \lshomo{$H$} in $P_t$-free graphs. We think it is interesting whether \lshomo{$P_3$} (and thus \lshomo{$C_4$}, as they are closely related) can be solved in subexponential time in $P_t$-free graphs.

\end{document}